\documentclass[11pt]{article}
\usepackage[margin = 1.2in]{geometry}
\usepackage{amsmath,amssymb,amsfonts,amsthm}
\usepackage{url,graphicx,tabularx,array,multirow}
\usepackage{natbib}

\usepackage{titlesec}
\titleformat{\section}{\sf\large}{\thesection.}{1em}{\filcenter}
\titleformat{\subsection}{\sf}{\thesubsection}{1em}{}
\titleformat{\subsubsection}{\sf \small}{\thesubsubsection}{1em}{\it }
\usepackage[titletoc,toc,title]{appendix}

\newcommand{\scX}{\mathcal{X}}
\newcommand{\bbR}{\mathbb{R}}
\newcommand{\dagg}[1]{{#1}^\dagger}

\newtheorem{theorem}{Theorem}
\newtheorem{lemma}[theorem]{Lemma}

\newtheorem{defn}{Definition}

\usepackage[linesnumbered,lined,boxed,commentsnumbered]{algorithm2e}

\SetKwFor{For}{for}{}{endfor}

\title{Joint Estimation of Quantile Planes over Arbitrary Predictor Spaces}
\author{Yun Yang and Surya Tokdar\\{\it University of California, Berkeley} and {\it Duke University}}
\date{}

\begin{document}

\maketitle
\begin{abstract}
In spite of the recent surge of interest in quantile regression, joint estimation of linear quantile planes remains a great challenge in statistics and econometrics. We propose a novel parametrization that characterizes any collection of non-crossing quantile planes over arbitrarily shaped convex predictor domains in any dimension by means of unconstrained scalar, vector and function valued parameters. Statistical models based on this parametrization inherit a fast computation of the likelihood function, enabling penalized likelihood or Bayesian approaches to model fitting. We introduce a complete Bayesian methodology by using Gaussian process prior distributions on the function valued parameters and develop a robust and efficient Markov chain Monte Carlo parameter estimation. The resulting method is shown to offer posterior consistency under mild tail and regularity conditions. We present several illustrative examples where the new method is compared against existing approaches and is found to offer better accuracy, coverage and model fit. 
\end{abstract}

\section{Introduction}
Quantile regression \citep[QR;][]{Koenker1978,Koenker2005} has recently gained increased recognition as a robust alternative to standard least squares regression, with applications to ecology, economics, epidemiology and climate science research \citep{Burgette2011, Elsner2008, Dunham2002, Abrevaya01}. By offering direct inference on the non-central parts of a response distribution, QR allows researchers to identify and quantify a wide range of regression heterogeneity where the predictors affect the quartiles or the tails of the response distribution differently than its mean or median. This is illustrated in Figure \ref{f:mcycle}(a), adapted from \citet{Koenker2005vig}, showing the estimated conditional quantile curves for the well-known motorcycle data \citep{Silverman1985} with ``Acceleration'' (head acceleration, in g) as the response and ``Time'' (time from impact, in ms) as the explanatory variable. The estimates do a much better job of capturing the complex relationship between the two variables than what could be inferred through a simple mean regression or from more modern nonparametric density regression techniques \citep{DeIorio2004,Tokdar2010} as shown in Figures \ref{f:mcycle}(b)-(c).

\begin{figure}[t!]
\centering
\begin{tabular}{ccc}
    \includegraphics[trim=6.5cm 10.8cm 7.5cm 9cm,clip=true,scale=0.65]{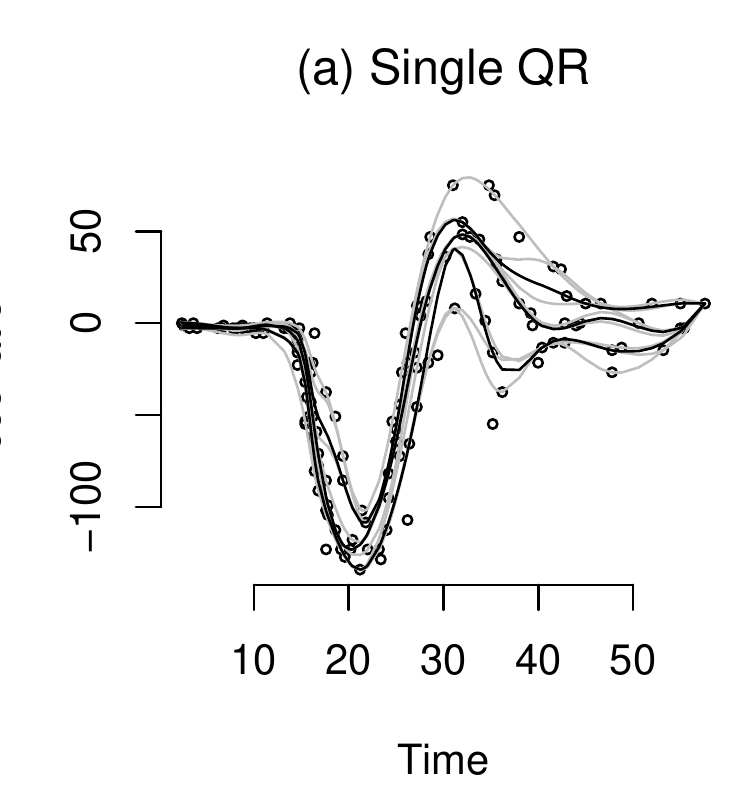}
    &
  \includegraphics[trim=6.5cm 10.8cm 7.5cm 9cm,clip=true,scale=0.65]{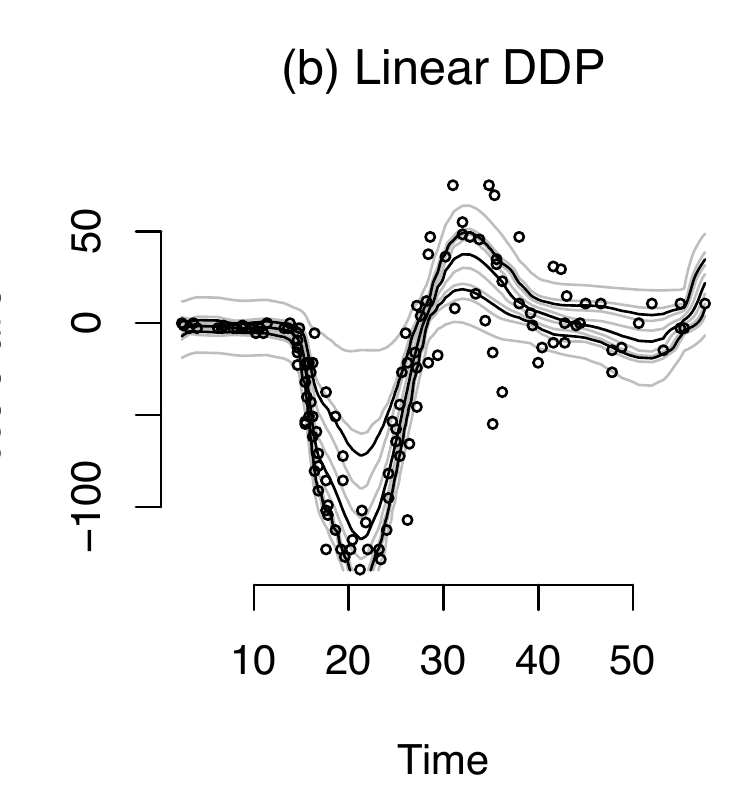}
  &
  \includegraphics[trim=6.5cm 10.8cm 7.5cm 9cm,clip=true,scale=0.65]{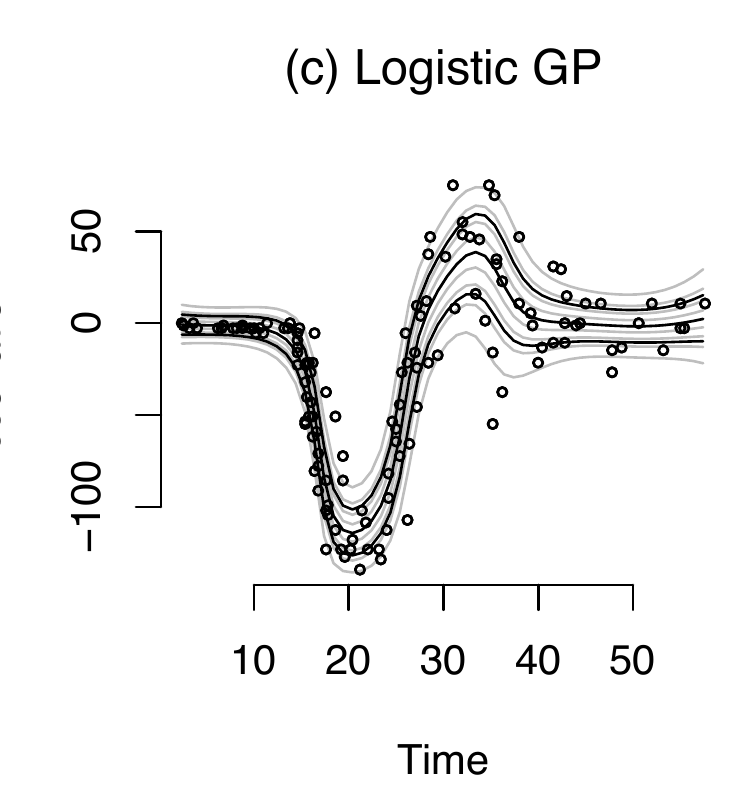}
  \\
    \includegraphics[trim=6.5cm 10.8cm 7.5cm 9cm,clip=true,scale=0.65]{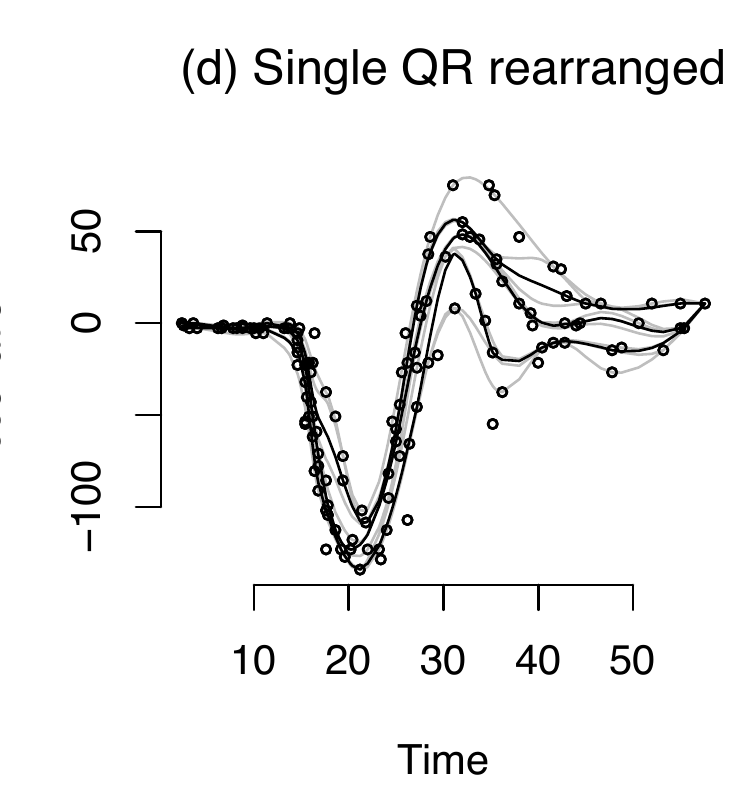}
  &
  \includegraphics[trim=6.5cm 10.8cm 7.5cm 9cm,clip=true,scale=0.65]{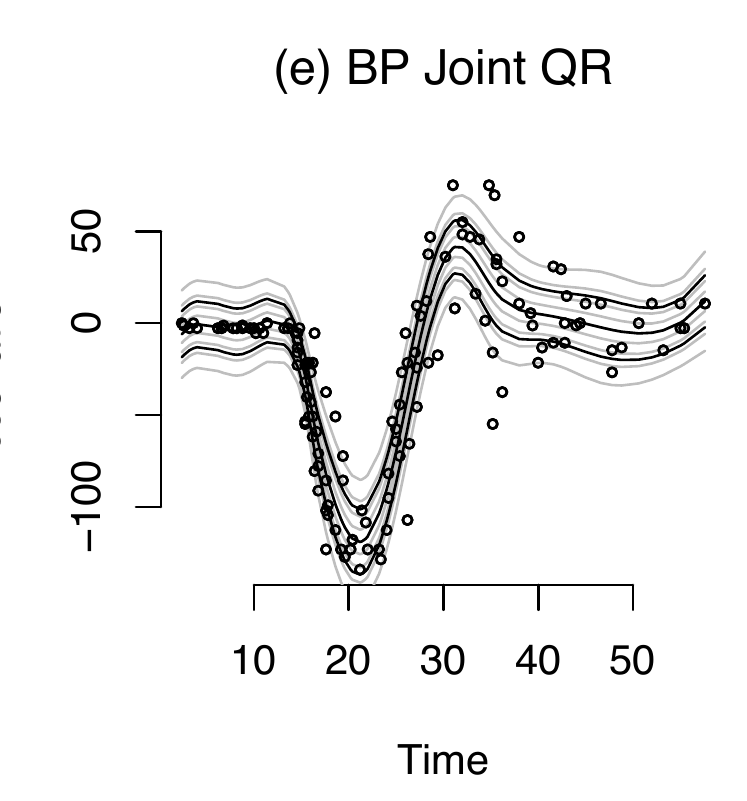}  
  &
  \includegraphics[trim=6.5cm 10.8cm 7.5cm 9cm,clip=true,scale=0.65]{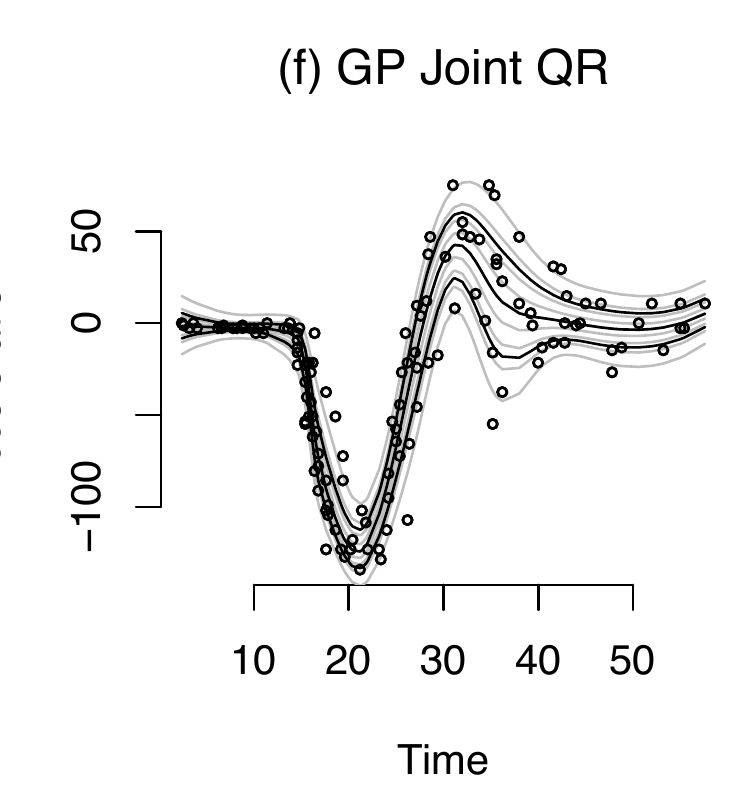}
\end{tabular}
\caption{Estimated quantile curves at $\tau \in \{0.1, 0.2, \cdots, 0.9\}$ (gray lines) and at $\tau \in \{0.25, 0.5, 0.75\}$ (black lines) for motorcycle data (open circles). Single QR fits were done with \texttt{rqss()} function of the \texttt{quantreg} R-pcakage. Rearrangement was done by obtaining single QR fits over the dense grid $\tau \in \{0.01, 0.02, \cdots, 0.99\}$. GP joint QR was from a preliminary implementation of the method described in Section 3.2. Linear DDP \cite{DeIorio2004} was implemented with the R-package \texttt{DPpackage} of \cite{DPpackage}.}
\label{f:mcycle}
\end{figure}

The estimates in \ref{f:mcycle}(a) were generated by using the original linear quantile regression technique of \citet{Koenker1978}. For a response proportion $\tau \in (0, 1)$ let $Q_Y(\tau | X) = \inf\{a: P(Y \le a | X) \ge \tau\}$ denote the $\tau$-th conditional quantile of a response $Y$ given a predictor vector $X$. The linear quantile regression model postulates
\begin{equation}
\label{eq1}
Q_Y(\tau | X) = \beta_0(\tau) + X^T\beta(\tau),
\end{equation}
which is equivalent to saying $Y = \beta_0(\tau) + X^T\beta(\tau) + U$ with the error variable satisfying $Q_U(\tau|X) \equiv 0$. The model is linear in the model parameters $(\beta_0(\tau), \beta(\tau))$. The predictor vector $X$ may include non-linear and interaction terms of the original covariates. In the motorcycle data analysis, we used B-spline transforms (df = 15) of Time as predictors, with $\dim(X) = 15$. The model parameters are easily estimated by linear programming and the estimates are consistent, asymptotically Gaussian and robust against outliers. Current literature on quantile regression (QR) is both deep and diverse; see \cite{Koenker2005} for a comprehensive overview and \cite{Tokdar2012} for references to Bayesian approaches.

Most scientific applications of QR require inference over a dense grid of $\tau$ values, which is usually done by assimilating inference from single-$\tau$ model fits \citep[e.g.,][]{Elsner2008}. Such assimilations are often problematic. In Figure \ref{f:mcycle}(a) the estimated curves cross each other violating laws of probability; the waviness and the local optima of the curves change wildly across $\tau$ reflecting poor borrowing of information; all quantile curves nearly collapse to a single point at boundary, where uncertainty should have been high due to data scarcity. Post-hoc rearrangement of the estimated quantiles \citep{Chernozhukov2011} avoids the embarrassing issue of crossing (Figure \ref{f:mcycle}(d)), but the other two problems persist.

Joint estimation of the conditional quantile planes requires working with the linear specification \eqref{eq1} simultaneously for all $\tau \in (0,1)$. These specifications together define a valid statistical model, parametrized by function valued parameters $\beta_0 : (0,1) \to \bbR$, $\beta: (0, 1) \to \bbR^p$, provided
\begin{equation}
\label{eq2}
\beta_0(\tau_1) + x^T \beta(\tau_1) \ge \beta_0(\tau_2) + x^T \beta(\tau_2),~\mbox{for every pair}~\tau_1 > \tau_2~\mbox{and for every}~x \in \scX
\end{equation}
where $\scX$ is a pre-specified domain for $X$. Such models and related methods are a minority in the current quantile regression literature, and existing approaches have severe shortcomings. The methods by \cite{He1997} and \cite{Bondell2010} impose serious restrictions on the shape of $\beta(\tau)$. The procedure by \cite{Dunson2005}, based on  substitution likelihood, does not scale to dense $\tau$ grids and the role of substitution likelihood in Bayesian estimation remains debated \citep{Monahan1992}. \cite{Tokdar2012} provide a complete, scalable solution for the univariate case, but their handling of multivariate $X$ through univariate, single index projection is unsatisfactory. 

To date, the most comprehensive treatment is given by \citet{Reich2011} who utilize monotonicity properties of Bernstein basis polynomials with non-negative coefficients to ensure non-crossing quantile planes in any dimension. Their use of truncated Gaussian prior distributions on the non-negative coefficients  leads to an attractive Gibbs sampling based Bayesian model fitting. However, both the model and the computing algorithm of \cite{Reich2011} crucially depend on the predictor domain $\scX$ being a hyper-rectangle in $\bbR^p$. This is a fairly major handicap that may lead to a poor fit for reasons explained below.

The specification of $\scX$ is a critical model choice in QR. Without loss of generality, $\scX$ can be chosen convex because \eqref{eq2} holds over $\scX$ if and only if it holds over the convex hull of $\scX$. The convex hull of the observed predictor vectors presents the most obvious practical choice. In spite of convexity, such an $\scX$ may have a fairly irregular shape and may occupy only a fraction of the volume of the encompassing hyper-rectangle. The B-spline transforms of Time in the motorcycle data analysis live on a tiny 1 dimensional manifold in $\bbR^{15}$. Quantiles planes that are required to be non-crossing over the larger hyper-rectangle will appear mostly parallel within the original $\scX$, as can be seen in Figure \ref{f:mcycle}(e). Unfortunately, such narrow predictor convex hulls are unavoidable whenever non-linear effects are sought within the Koenker-Bassett program or the measured covariates are naturally correlated. These are also situations where assimilation techniques exhibit dramatic crossing problems and hence a sound statistical model is most needed for joint estimation.

For an arbitrary convex $\scX$, the space of $\beta_0$, $\beta$ curves satisfying \eqref{eq2} is highly non-regular and unsuitable for statistical modeling and investigation. For the case of $p = 1$, \cite{Tokdar2012} provides a much simpler representation parametrized by two monotonically increasing curves over $(0,1)$. A generalization of this to any $p \ge 1$ and any convex $\scX$ of arbitrary shape is currently not available in the literature. 

In this paper we propose a novel theory that delivers the right modeling platform for joint quantile regression. Our theory covers any dimension $p$ and any bounded  convex $\scX$ of arbitrary shape. It provides a complete characterization of joint quantile regression in terms of a collection of scalars, vectors and curves all but one of which are entirely constraint-free. Even the one curve with a constraint has only a mild shape restriction on it; it is required to live in the space of all CDFs on $(0,1)$ with full support. Our reparametrization leads to an easy likelihood score calculation in the model parameters, making it ideally suited to develop practicable methods by using either penalized likelihood or Bayesian techniques.

We build upon this novel theory to introduce a semiparametric Bayesian methodology for joint quantile regression over any $\scX$ and any $p$, where the curve valued model parameters are assigned Gaussian process and transformed Gaussian process priors within a hierarchical setting. Asymptotic frequentist properties of the method are studied in Section \ref{s:theory} and we establish posterior consistency over a broad class of true data generating distributions with linear quantile curves.
For parameter estimation, we propose a Monte Carlo technique that incorporates efficient model space discretization, adaptive Markov chain sampling \citep{Haario1999} and reduced rank approximation \citep{tokdar07, Banerjee2008}. We provide empirical evidence (Section \ref{s:simu1}-\ref{s:cs}) that our Gaussian process method enjoys much better estimation accuracy and coverage than the method by \citet{Reich2011}, and our estimates are comparable to regularized versions of the classical single-$\tau$ estimates. We consider the developments here make a strong case for linear quantile regression to be used as a model based inferential method rather than just an exploratory tool!

\relax
\section{A novel theory of joint quantile planes estimation}
\label{s:necsuff}
\subsection{Characterizing non-crossing hyperplanes}

We focus only on the case where the response distribution is non-atomic and admits a probability density function conditionally at every $X$, and hence \eqref{eq2} is equivalent to requiring $\dot\beta_0(\tau) + x^T\dot\beta(\tau) > 0$ for all $\tau \in (0,1), x \in \scX$. Our theory could be extended to atomic response distributions with known atoms. Assume $0$ is an interior point of $\scX$. This can be achieved without any loss of generality by a simple translation of the predictors once a suitable interior point is found within the convex hull of the observed predictors; see Appendix \ref{A:centering} for more details. Define a map $b \mapsto a(b, \scX)$ on $\bbR^p\cup \{\infty\}$ as
\[
a(b, \scX) = \left\{\begin{array}{ll} \sup_{x \in \scX} \{-x^Tb\}/\|b\| & b \ne 0,\\\infty & b = 0.\end{array}\right.
\]
Note that for every $b \ne 0$ we have $a(b, \scX) \in (0, \infty)$ because $\scX$ is bounded with 0 as an interior point.

\begin{theorem}
\label{thm:char}
Let $\scX$ be a bounded convex set in $\bbR^p$ with zero as an interior point and let $\beta_0(\tau)$ and  $\beta(\tau) = (\beta_1(\tau), \cdots, \beta_p(\tau))^T$ be real, differentiable functions in $\tau \in (0, 1)$. Then $\dot\beta_0(\tau) + x^T\dot\beta(\tau) > 0$ for all $\tau \in (0, 1)$ at every $x \in \scX$ if and only if
\begin{align}
\dot\beta_0(\tau) > 0,~~\dot\beta(\tau) = \dot\beta_0(\tau) \frac{v(\tau)}{a(v(\tau), \scX) \sqrt{1 + \|v(\tau)\|^2}},~~\tau \in (0,1), \label{c2}
\end{align}
for some $p$-variate, real function $v(\tau) = (v_1(\tau), \cdots, v_p(\tau))^T$ in $\tau \in (0, 1)$.
\end{theorem}

\begin{proof}
\underline{If part.} Suppose \eqref{c2} holds. For any $\tau \in (0,1)$ either $v(\tau) = 0$ in which case $\dot\beta(\tau) = 0$ and so $\dot\beta_0(\tau) + x^T\dot\beta(\tau) > 0$ at every $x \in \scX$. Otherwise, if $v(\tau) \ne 0$ then at any $x \in \scX$.
\begin{align*}
\dot\beta_0(\tau) + x^T\dot\beta(\tau) & = \dot\beta_0(\tau)\left\{1 + \frac{x^T v(\tau)}{a(v(\tau), \scX) \sqrt{1+ \|v(\tau)\|^2}} \right\}\\
& = \dot\beta_0(\tau) \left\{1 - \sqrt{\frac{\|v(\tau)\|^2}{1 + \|v(\tau)\|^2} } \frac{\{-x^T v(\tau) \}/ \|v(\tau)\|}{a(v(\tau), \scX)}\right\}\\
& \ge \dot\beta_0(\tau)\left\{1 - \sqrt{\frac{\|v(\tau)\|^2}{1 + \|v(\tau)\|^2} } \right\}\\
& > 0
\end{align*}

\noindent \underline{Only if part.} We must have $\dot\beta_0(\tau) > 0$ for all $\tau \in (0, 1)$ because $\scX$ contains 0.  For any $\tau \in (0, 1)$, if $\dot\beta(\tau) = 0$, set $v(\tau) = 0$. Otherwise, $\dot\beta_0(\tau) > \{-x^T\dot \beta(\tau)\}$ at every $x \in \scX$ and hence $\dot\beta_0(\tau) > \|\dot\beta(\tau)\|a(\dot\beta(\tau), \scX)$. So the positive scalar  $c(\tau) = [ [\dot\beta_0(\tau)/\{\|\dot\beta(\tau)\|a(\dot\beta(\tau), \scX)\}]^2  - 1]^{-1/2}$ satisfies $\|\dot\beta(\tau)\| a(\dot\beta(\tau), \scX) / \dot\beta_0(\tau) = c(\tau) / \{1  + c(\tau)^2\}^{1/2}$. Set
\[v(\tau) = c(\tau)\dot\beta(\tau) / \|\dot \beta(\tau)\|,~~\tau \in (0, 1).
\]
A $v(\tau)$ constructed as above defines a real $p$-variate function on $\tau \in (0, 1)$ and satisfies \eqref{c2}.
\end{proof}

\subsection{An almost constraint-free parametrization of linear quantile regression}

Theorem \ref{thm:char} greatly reduces the monotonicity constraint on the quantile hyperplanes to that on a single function $\beta_0(\tau)$. Construction of a single monotone function is a relatively easy task,
but some care is needed in handling the range $(\beta_0(0), \beta_0(1))$, which corresponds to the support of the conditional density of $Y$ given $X = 0$. We pursue a model for $\beta_0(\tau)$ based on a user specified or default ``prior guess'' $f_0(y)$ for this conditional density.
In the special case where $(\beta_0(0), \beta_0(1))$ is a known finite interval, $f_0$ could be chosen with support equal to the same interval. In general $f_0$ should be chosen to have support $(-\infty, \infty)$, such as a standard normal density, or a Student-t density with a modest degrees of freedom if the response distribution is expected to have heavy tails. We focus only on this general case, although the model described below could be easily modified to $f_0$ supported on a bounded interval.

Let $f_0$ have support $(-\infty, \infty)$ and define its cumulative distribution function $F_0(y) = \int_{-\infty}^y f_0(z)dz$, quantile function $Q_0(\tau) = F_0^{-1}(\tau)$ and quantile density $q_0(\tau) = \dot Q_0(\tau)$. Let $\tau_0 = F_0(0)$; by the full support assumption, $0 < \tau_0 < 1$. We pursue a model for $\beta_0$ and $\beta$ as follows
\begin{align}
& \beta_0(\tau_0) = \gamma_0,~~\beta(\tau_0) = \gamma \label{eq:tau0}\\
& \beta_0(\tau) - \beta_0(\tau_0) = \sigma \int_{\zeta(\tau_0)}^{\zeta(\tau)} q_0(u) du, ~~~ \tau \in (0,1) \label{eq:beta0}\\
& \beta(\tau) - \beta(\tau_0) = \sigma \int_{\zeta(\tau_0)}^{\zeta(\tau)} \frac{w(u)}{a(w(u), \scX)\sqrt{1 + \|w(u)\|^2}}q_0(u)du,~~~ \tau \in (0,1) \label{eq:beta}
\end{align}
with model parameters $\gamma_0 \in \bbR$; $\gamma \in \bbR^p$; $\sigma > 0$; $w : (0, 1) \to \bbR^p$, an unconstrained $p$-variate function on $(0,1)$; and $\zeta : [0, 1] \to [0, 1]$, a differentiable, monotonically increasing bijection, i.e., a diffeomorphism, of $[0,1]$ onto itself. We write $(\beta_0, \beta) = \mathcal{T}(\gamma_0, \gamma, \sigma, w, \zeta)$ to indicate $\beta_0, \beta$ defined as in \eqref{eq:tau0}-\eqref{eq:beta}. 

All model parameters, except the diffeomorphism $\zeta$, are essentially unconstrained. The function space of $\zeta$ is simply the space of cumulative distribution functions associated with all probability densities with support $[0,1]$. Such function spaces are easy to handle for statistical model fitting; a simple approach is presented in Section 3. Note that when $\zeta$ is the identity map of $(0,1)$ onto itself and $w(\tau) \equiv 0$, we get $\beta_0 = \sigma Q_0$, $\beta = 0$, and the resulting joint, linear quantile regression model simplifies to a standard homogeneous, linear model: $Y_i = \gamma_0 + X_i^T\gamma + \sigma \epsilon_i$ with $\epsilon_i \sim f_0$. This model indeed provides a complete representation of all $\beta_0$, $\beta$ satisfying the non-crossing condition \eqref{eq2}, subject to a matching range criterion, as detailed below.
\begin{theorem}
\label{thm:model}
Let $\beta_0: (0,1) \to \bbR$, $\beta: (0,1) \to \bbR^p$ be differentiable with $(\beta_0(0), \beta_0(1)) = (-\infty, \infty)$ [defined in the limit]. Then \eqref{eq2} holds if and only if $(\beta_0, \beta) = \mathcal{T}(\gamma_0, \gamma, \sigma, w, \zeta)$ for some $\gamma_0 \in \bbR$, $\gamma \in \bbR^p$, $\sigma > 0$, $w : (0,1) \to \bbR^p$, and, $\zeta$, a diffeomorphism from $[0,1]$ onto itself $[0,1]$. 
\end{theorem}

\begin{proof}
If $(\beta_0,\beta) =  \mathcal{T}(\gamma_0, \gamma, \sigma, w, \zeta)$ then,
\[
\dot\beta_0(\tau) = \sigma q_0(\zeta(\tau)) \dot\zeta(\tau), ~~\dot\beta(\tau) = \dot\beta_0(\tau) \frac{v(\tau) }{a(v(\tau), \scX) \sqrt{1 + \|v(\tau)\|^2}},
\] 
with $v(\tau) = w(\zeta(\tau))$. Hence, by Theorem \ref{thm:char}, we only need establish that any real, differentiable function $\beta_0$ on $(0,1)$, with $\dot\beta_0(\tau) > 0$ for all $\tau \in (0,1)$ and $\beta_0(0) = \infty$, $\beta_0(1) = \infty$, can be constructed as in \eqref{eq:beta0} for some diffeomorphism $\zeta:[0,1]\to[0,1]$ and $\sigma > 0$. This is indeed true, since one could fix $\sigma > 0$ arbitrarily, and then take,
\[
\zeta(\tau) = F_0\left(\gamma_0 + \frac{\beta_0(\tau)-\gamma_0}\sigma \right),~~\tau \in (0,1),
\]
which is differentiable and monotonically increasing in $(0,1)$ since $\dot\zeta(\tau) = f_0(\gamma_0 + (\beta_0(\tau) - \gamma_0)/\sigma) \dot\beta_0(\tau) > 0$ for all $\tau \in (0,1)$, and, $\zeta(0) = F_0(-\infty) = 0$, $\zeta(1) = F_0(\infty) = 1$.
\end{proof}

When $\beta_0(0)$ or $\beta_0(1)$ is finite (or both), we can still write $\beta_0$ as in \eqref{eq:beta0}, but we need either $\zeta(0) > 0$ or $\zeta(1) < 1$ (or both). While such a $(\beta_0, \beta)$ does not strictly belong within our model space, they can be approximated arbitrarily well by a model element $\mathcal{T}(\gamma_0, \gamma, \sigma, v, \zeta)$, and consistently estimated from large samples (Lemma \ref{le:6} and Section \ref{s:theory}).



\subsection{Likelihood evaluation}
\label{ss:le}
%
A salient feature of a valid specification of $Q_Y(\tau|x)$ for all $\tau \in (0,1)$ is that it uniquely defines the conditional response density $f_Y(y|x)$ over $x\in \scX$, given by
\[
f_Y(y|x) = \left.\frac{1}{\frac{\partial}{\partial\tau} Q_Y(\tau | x)}\right|_{\tau = \tau_x(y)}
\]
where $\tau_x(y)$ solves $Q_Y(\tau | x) = y$ in $\tau$ \citep{Tokdar2012}.
Consequently, one can define a valid log-likelihood score
\begin{equation}\label{eq8}
\sum_{i}\log f_Y(y_i|x_i)=-\sum_i\log\big\{\dot{\beta}_0\big(\tau_{x_i}(y_i)\big)+
 x_i^T\dot{\beta}\big(\tau_{x_i}(y_i)\big)\big\}.
\end{equation}
in the model parameters based on observations $(x_i, y_i)$, $i = 1, \ldots, n$. From \eqref{eq:tau0}-\eqref{eq:beta}, we could write 
\[
\dot \beta_0(\tau) + x^T \dot \beta(\tau) = \sigma q_0(\zeta(\tau))\dot \zeta(\tau) \left\{1 + \frac{w(\zeta(\tau))}{a(w(\zeta(\tau), \scX)\sqrt{1+ \|w(\zeta(\tau))\|^2}}\right\}
\]
and therefore a quick evaluation of the log-likelihood score is possible once we figure out $\tau_{x_i}(y_i)$ for each $i = 1, \ldots, n$, by solving $\tau = \int_{\tau_0}^\tau \{\dot\beta_0(u) + x_i^T \dot \beta(u)\}du$. 

With enough resources, these numbers could be found up to any desired level of accuracy through standard numerical methods for integration and root finding. But for all practical needs, model fitting and inference could be restricted to a dense grid of $\tau \in \{t_1, \ldots, t_L\}$, for which finding $\tau_{x_i}(y_i)$ requires only a simple sequential search involving trapezoidal approximations to the integral of $\dot\beta_0(\tau) + x^T\dot\beta(\tau)$. Algorithm \ref{algo:loglik} presents a pseudo-code for likelihood evaluation involving only simple matrix and vector multiplication. The code runs extremely fast when implemented in any low-level programming language with quick ``for loops''. In our numerical studies we used a C implementation which offered 1000 likelihood evaluations in 2 seconds on an Intel(R) Core(TM) i7-3770 machine with $n = 1000$, $p = 7$ and a grid over $\tau$ with mesh size 0.01.

A practical issue with a discrete grid of $\tau$ is that it needs to cover the image of the data range mapped into the quantile space, while ensuring the grid length $L$ remains manageable. In our implementations we chose a data dependent grid as follows. We used equispaced grid points between $\tau = 0.01$ and $\tau = 0.99$ with an increment of $0.01$. Next, on the upper tail, we augmented the grid with new grid points $0.995$, $0.9975$, $\cdots$ until we covered $\tau = 1 - 1/(2n)$ where $n$ is the sampler size. Same augmentation strategy with geometrically reducing increment lengths were adopted on the lower tail to reach up to $\tau = 1/(2n)$. 

\relax
\section{Bayesian inference with hierarchical Gaussian process priors}

\label{s:prior}

\subsection{Prior specification}

We adopt a Bayesian approach to parameter estimation with suitable prior distributions on the model parameters, including the function valued parameters $\zeta$ and $w = (w_1, \ldots, w_p)$. It is useful that $w_j$s are completely unrestricted, allowing us to handle them with Gaussian process prior distributions. For handling $\zeta$, we first introduce a constraint free version $w_0:(0,1) \to \bbR$ related to $\zeta$ through the ``logistic transformation'':
\begin{equation}
    \zeta(\tau)=\frac{\int_0^\tau e^{w_0(u)}du}{\int_0^1 e^{w_0(u)}du}, \;\;\tau \in (0,1),
    \label{eq:zeta0}
\end{equation}
and use a Gaussian process prior on $w_0$; see \cite{Lenk88, tokdar07} for similar uses in density estimation. 

Recall that a Gaussian process $g = \{g(\tau): \tau \in (0,1)\}$ could be viewed as a random element of the Banach space of real valued functions on $(0,1)$ equipped with the supremum norm. Every Gaussian process $g$ is characterized by two functions, the mean function $m(\tau) = {\rm E} g(\tau)$ and the non-negative definite covariance function $c(\tau, \tau') = {\rm Cov}(g(\tau), g(\tau'))$, and we use the label $GP(m, c)$ to denote such a process. When $g \sim GP(m, c)$, for any finite set of points $\{\tau_1, \ldots, \tau_k\}$ the random vector $(g(\tau_1), \ldots, g(\tau_k))$ has a $k$-variate Gaussian distribution with mean $(m(\tau_1), \cdots, m(\tau_k))^T$ and $k\times k$ covariance matrix with elements $c(\tau_i, \tau_j)$.

Our prior specification can be expressed in the following hierarchical form:
\begin{align}
w_j & \sim GP(0, \kappa_j^2 c^{\textit{\tiny SE}}(\cdot,\cdot|\lambda_j)),~~j = 0,\ldots,p \label{eq:w}\\
(\kappa_j^2, \lambda_j) & \sim \pi_k(\kappa_j^2)\pi_\lambda(\lambda_j), ~~ j = 0,\ldots, p\\
(\gamma_0, \gamma, \sigma^2) & \sim \pi(\gamma_0,\gamma,\sigma^2) \propto \frac1{\sigma^2},
\end{align}
where $c^{\textit{\tiny SE}}(\tau, \tau'|\lambda^2) =  \exp(-\lambda^2 (\tau - \tau')^2)$ is the so-called square exponential covariance function equipped with a rescaling parameter $\lambda$ \citep{vanderVaart&vanZanten08}. This particular choice of the covariance function is motivated by two facts. First, for any fixed $\lambda > 0$, the probability distribution $GP(0, c^{\textit{\tiny SE}}(\cdot,\cdot|\lambda))$ assigns 100\% probability to the set of all continuous functions on $(0,1)$ and hence our prior specification does not {\it a-priori} rule out any valid specification of the joint linear QR model. Second, $\lambda$ plays the role of a bandwidth parameter for the sample paths generated from $GP(0, c^{\textit{\tiny SE}}(\cdot,\cdot|\lambda))$, with more wavy paths realized as $\lambda$ gets larger. In a seminal work, \cite{vanderVaart&vanZanten09} show that with a suitable prior distributions specified on $\lambda$, the resulting rescaled square-exponential Gaussian process prior offers adaptively efficient estimation in nonparametric mean regression and density estimation problems by automatically adjusting $\lambda$ to attain optimal smoothing.

For specifying $\pi_\lambda$, it is more insightful to fix a small $h > 0$ and consider the quantity $\rho_h(\lambda) = \exp(-h^2\lambda^2)$, which gives the correlation between $w_j(\tau)$ and $w_j(\tau + h)$ given $\lambda_j = \lambda$, and assign $\rho_h(\lambda)$ a $Be(a_\lambda,b_\lambda)$ prior. In our applications we use $h = 0.1$, $a_\lambda = 6$ and $b_\lambda = 4$, which assigns 95\% mass to $\rho_{0.1}(\lambda) \in (0.3, 0.86)$. However, in our experience, the method shows little sensitivity to these choices. We take $\pi_\kappa$ to be $IG(a_\kappa, b_\kappa)$, the inverse gamma pdf with shape $a_\kappa$ and rate $b_\kappa$. The inverse gamma choice allows us to integrate out all $\kappa_j$ parameters at the time of model fitting. In our applications, we use $a_\kappa = b_\kappa = 3/2$, which is small enough to ensure a reasonably diffuse marginal prior on each $w_j$ while retaining a finite second moment. 

Our choice of $\pi_\kappa$ and the right Haar prior on the location scale parameters $(\gamma_0, \gamma, \sigma^2)$ is partially motivated by our numerical experimentations in which we found these choices to lead to estimates and credible intervals most similar to the Koenker-Basette estimates and confidence intervals. Other reasonable choices could be made and we discuss in Section \ref{s:dis} choices that offer useful shrinkage properties. 

When no special information is available about the support of $Y$, we take $f_0$ to be a Student t-distribution with an unknown degrees of freedom parameter $\nu$ and assign $\nu/6$ a standard logistic prior distribution. The logistic prior is reasonably diffuse and helps the resulting method adapt well to a wide spectrum of tail behavior of the response distribution.

\subsection{Model fitting via discretization and adaptive blocked Metropolis} 
\label{s:comp}

With likelihood evaluation discretized over a grid of $\tau$ values $\{t_1, \ldots, t_L\}$ as in Algorithm \ref{algo:loglik}, the curve valued parameters $w_j$, $1\le j \le p$ are needed to be tracked only over the specified grid, reducing each curve to a parameter vector of length $L$. The same applies to $w_0$ from which $\dot\zeta$ and $\zeta$ could be obtained on the grid by using the trapezoidal rule of integration. While it is theoretically possible to fit the model by running a Markov chain Monte Carlo over these parameters vector and the other model parameters, such a strategy is not entirely practicable. The parameter vector derived from any $w_j$ is conditionally an $L$ dimensional Gaussian variable given $\lambda_j$ and $\kappa_j$, and evaluating its log prior density requires factorizing or inverting a $L\times L$ covariance matrix which has an $O(L^3)$ computing complexity. Furthermore, a Markov chain sampler that operates on both these parameter vectors and the rescaling parameters $\lambda_j$s run into serious mixing problems. 

To overcome these difficulties, we use two sets of further discretization. First, we replace $\pi_\lambda$ with a dense, discrete approximation covering the range $\rho_{0.1}(\lambda) \in (0.05, 0.95)$. Let $\pi^*_\lambda$ denote the approximating probability mass function with support points $\{\lambda^*_1, \ldots, \lambda^*_G\}$. We choose the support points to be more densely packed for smaller $\lambda$ values, the rational behind this and the exact manner in which the grid is chosen are discussed in Appendix \ref{a:support}. 

Next, we fix a set of uniformly spaced knots $\{t^*_1, \ldots, t^*_m\} \subset[0,1]$, for some $m$ much smaller than $L$ and replace each $w_j$ curve with
\begin{equation}
\tilde w_j(\tau) := E\{w_j(\tau) | w_j(t^*_1), \ldots, w_j(t^*_m)\}, \tau \in (0,1),
\label{eq:pp}
\end{equation}
which provides an interpolation approximation to $w_j$ over $(0,1)$, passing through the points $(t^*_k, w_j(t^*_k))$, $k = 1, \ldots, m$, and determined entirely by the $m$-dimensional vector $W_{j*} = (w_j(t^*_1), \ldots, w_j(t^*_m))^T$, whose prior density evaluations require only $O(m^3)$ flops. Such interpolation based low rank approximations to Gaussian process priors are widely used in statistics and machine learning literature, see for example, \citet{snelson&ghahramani06, tokdar07, banerjee&etal08}. 

Our treatment here, however, differs slightly from the above papers in that we carry out the conditional expectation in \eqref{eq:pp} after marginalizing out both $\lambda_j$ and $\kappa_j$. Let $\tilde W_j$ denote the $L$-dimensional vector $(\tilde w_j(t_1), \ldots, \tilde w_j(t_L))^T$ that is needed for the likelihood evaluation. Then we can write,
\[
\tilde W_j = \sum_{g = 1}^G p_g(W_{j*})A_g W_{*j}
\]
where $A_g$ denotes the $L\times m$ matrix $C_{o *}(\lambda_g)C_{**}(\lambda_g)^{-1}$ with $C_{o*}(\lambda_g) = ((c^{\textit{\tiny SE}}(t_l, t^*_k|\lambda_g)))_{l,k=1}^{L,m}$ and $C_{**}(\lambda_g) = ((c^{\textit{\tiny SE}}(t^*_l, ^*t_k|\lambda_g)))_{l,k=1}^{m}$, and $p_g(W_{j*}) \propto \pi^*_\lambda(\lambda_g)p(W_{j*} | \lambda_g)$ with
\[
p(W_{j*} | \lambda_g) \propto \pi^*_\lambda(\lambda_g)\left\{1 + \frac{W_{j*}^TC^{-1}_{**}(\lambda_g)W_{j*} }{2b_\kappa}\right\}^{-(a_\kappa + m/2)} \frac{\Gamma(a_\kappa + m/2)b_\kappa^{-m/2}}{\Gamma(a_\kappa)},
\]
the multivariate t-density of $W_{j*}$ given $\lambda_j = \lambda_g$. Also notice that the marginal prior density of $W_{j*}$ is precisely $\sum_{g = 1}^G \pi^*_\lambda(\lambda_g) p(W_{j*}|\lambda_g)$.

With the help of the above sets discretization, our joint QR model is entirely determined by the $(m+1)(p+1)+2$ dimensional parameter vector $\theta = (W_{0*}^T, \ldots, W_{p*}^T, \gamma_0, \gamma^T, \sigma^2, \nu)^T$ and model fitting may be carried out by running a Markov chain sampler on $\theta$ followed by Monte Carlo approximations of posterior quantities. In our experience, an adaptive blocked Metropolis sampler has worked extremely well, offering fast mixing and reproducible results. For this sampler, we use $p+3$ block updates of $\theta$ per iteration of the sampler, where the first $p+1$ blocks are given by $(W_{j*}^T, \gamma_j)^T$, $j = 0,\ldots,p$ and the last two blocks are $(\gamma_0, \gamma^T)^T$ and $(\log \sigma^2, \log \nu)^T$. For each block, we perform a random walk Metropolis update governed by a multivariate Gaussian proposal distribution centered at the current realization of the block and with covariance that is slowly adapted to resemble, up to a scaler multiplication, the posterior covariance matrix of the block, where the scaler multiplier is also adapted slowly to achieve a pre-specified acceptance rate. We carry out these updates according to Algorithm 4 in \citet{andrieu.thoms}.

In our implementation, we precompute and save the matrix $A_g$ and a Cholesky factor $R_g$ of $C_{**}(\lambda_g)$ for every $g = 1, \ldots, G$ and plug them into the likelihood and prior density evaluations during Markov chain sampling. The precomputation step adds little overhead cost but  results in a big jump in computing speed by drastically reducing the computing time for each Markov chain iteration.


\relax
\section{Posterior consistency}
\label{s:theory}

Frequentist justification of Bayesian methods are often presented in the form asymptotic properties of the posterior distribution. A basic desirable property is posterior consistency: the posterior mass assigned to any fixed neighborhood of the true data generating model element should converge to 1 in probability or almost surely as sample size goes to infinity. More refined evaluations of asymptotic properties emerge through posterior convergence rate calculations, where one considers a sequence of shrinking neighborhoods and calibrates the fastest rate of shrinkage for which the posterior mass assigned to these neighborhoods still converges to 1. 

We restrict only to a study of {\it weak} posterior consistency of the Gaussian process based QR method developed in this paper. For a formal treatment, we consider a stochastic design setting where $X_i$s are drawn independently from a pdf $f_X$ on $\scX$. Since any valid specification of the quantile planes $\{Q_Y(\tau|x)$: $\tau \in (0,1)$, $x \in \scX\}$ uniquely corresponds to a specification of conditional response densities $\{f_Y(y|x) : y\in \bbR, x \in \scX\}$, it also uniquely corresponds to a bivariate density function $f(x, y) = f_X(x)f_Y(y|x)$ under the stochastic design assumption. Hence our prior specification on the quantile planes induces a prior probability measure $\Pi$ on the space $\mathcal{F}$ of probability density functions on $\scX \times \bbR$. If $f^*(x,y) = f_X(x)f^*_Y(y|x)$ is the true data generating element in this space, then the posterior is said to be weakly consistent at $f^*$ if $\Pi(U|(X_i,Y_i), i = 1, \ldots, n) \to 1$ almost surely for every weak neighborhood $U$ of $f^*$ in $\mathcal{F}$. 

The celebrated Schwartz Theorem \citep{Schwartz1965} provides a fairly sharp sufficient condition for weak posterior consistency of $\Pi$ at $f^*$. Let $d_{KL}(p,q) := \int p \log (p/q)$ denote the Kullback-Leibler (KL) divergence. For any $f\in\mathcal{F}$ and $\epsilon > 0$, let $K_{\epsilon}(f)$ denote the $\epsilon$-KL neighborhood $\{g\in\mathcal{F}: d_{KL}(f,g)<\epsilon\}$. We say that $f^*$ is in the KL support of $\Pi$ if $\Pi(K_{\epsilon}(f^*))>0$ for all $\epsilon > 0$. \citet{Schwartz1965} proved
\begin{theorem}[Schwartz]\label{thm:2}
The posterior is weakly consistent at $f^* \in \mathcal{F}$ if $f^*$ is in the KL support of $\Pi$.
\end{theorem}

We show that an $f^*$ with linear conditional quantiles $Q^*_Y(\tau|x) = \beta^*_0(\tau) + x^T\beta^*(\tau)$ belongs to the KL support of $\Pi$ under mild smoothness and tail conditions. Tail conditions are needed to ensure that $d_{KL}(f^*_Y(\cdot|x), f_Y(\cdot|x)) < \infty$, which holds when $f^*_Y(\cdot|x)$ has tails decaying faster than those of $f_Y(\cdot|x)$, with $f$ generated from $\Pi$. With our choice of $\Pi$, the tails of $f_Y(\cdot|x)$ are expected to be similar to those of $f_0$, and hence, a minimum requirement is that the tails of $f^*_Y(\cdot|x)$ decay faster than those of $f_0$. We make the notion of faster tail decay more precise with the following definitions. 
\begin{defn}
Let $f$ be a probability density function on $\bbR$ with quantile function $Q$. Take $m = Q(\tau_0)$. All statements below are interpreted with respect a given $f_0$.
\begin{enumerate}
\item We say $f$ has a type I left tail if $Q(0) > -\infty$, and, for every $\sigma > 0$, 
\begin{equation}
\frac{\frac1\sigma f_0(m + \frac{Q(t) - m}\sigma)}{f(Q(t))} \to c_L(\sigma) \in (0, \infty), ~\mbox{as}~t \downarrow 0,
\label{eq:type1}
\end{equation} 
with, $c_L(\sigma) \to 0$ as $\sigma \downarrow 0$.
\item We say $f$ has a type II left tail if for every $\sigma > 0$, ${\frac1\sigma f_0(m + \frac{Q(t) - m}\sigma)}/{f(Q(t))} $ diverges to $\infty$ as $t \downarrow 0$ and,
\begin{equation}
u_L(\sigma) :=  \inf\left\{t > 0: \frac{\frac1\sigma f_0(m + \frac{Q(t) - m}\sigma)}{f(Q(t))} \le 1\right\} > 0,
\label{eq:type2}
\end{equation}
with, $u_L(\sigma) \to 0$ as $\sigma \downarrow 0$.
\item Same definitions apply to the right tail, with, $Q(1 - t)$ replacing $Q(t)$ in \eqref{eq:type1}, \eqref{eq:type2}, and $c_R$ and $u_R$ denoting the right tail counterparts of $c_L$ and $u_L$.
\end{enumerate}
\end{defn}
Recall that we have taken $f_0 = f_0(\cdot|\nu) = t_\nu$ with a prior on $\nu \in (0,\infty)$. Notice that an $f$ has a type I left tail with respect to any $t_\nu$, when ${\rm supp}(f)$ is bounded from below, which is same as saying $Q(0) > -\infty$, and, $f(y)$ is bounded away from zero near $Q(0)$. If $Q(0) > -\infty$ but $f(y) \to 0$ as $y \to Q(0)$ then $f$ has a type II left tail with respect to any $t_\nu$. If $Q(0) = -\infty$ and $f(y)$ decays to zero as $y \to -\infty$ at a polynomial or faster rate, then, $f$ has a type II left tail with respect to $t_\nu$ for all $\nu > 0$ sufficiently small. It is straightforward to see that $d_{KL}(f, f_0) < \infty$ whenever $f$ has tails that are type I or type II with respect to $f_0$.

It turns out that a type I or II tail condition on $f^*_Y(\cdot|0)$, coupled with some regularity conditions on $\beta_0,\beta$ are all that is needed to ensure consistency. Here is a precise statement.

\begin{theorem}\label{th:2}
Suppose $\beta^*_0$, $\beta^*$ are differentiable on $(0,1)$. Also assume $\dot\beta^* / \dot\beta^*_0$ can be extended to a continuous function on $[0,1]$, and, there exists a $c_0 > 0$ such that $\dot\beta^*_0(t) + x^T\dot\beta^*(t) \ge c_0 \dot\beta^*_0(t)$ for all $t \in (0,1)$. Then $f^*$ belongs to the KL support of $\Pi$ whenever $f^*_Y(\cdot|0)$ has type I or II tails with respect to $t_\nu$ for all small enough $\nu > 0$.
\end{theorem}
A proof is given in Appendix \ref{a:proof4}. The two regularity conditions on $(\dot\beta^*_0, \dot\beta^*)$ ensure that the conditional density functions do not exhibit pathological behaviors in the tails. Notice that the basic validity assumption $\dot\beta^*_0(t) + x^T\dot\beta^*(t) > 0$ for all $t \in (0,1)$ automatically guarantees that $\dot\beta^*(t)/\dot\beta^*_0(t)$ is bounded for all $t$. To see this, notice that $\scX$ must contain an open ball of radius $r > 0$ around origin which is an interior point. So, for any $t \in (0,1)$ with $\dot\beta^*(t) \ne 0$, $u := -r\dot\beta^*(t)/\|\dot\beta^*(t)\| \in \scX$, and hence, $0 \le \dot\beta^*_0(t) + u^T\dot\beta^*(t) = \dot\beta^*_0(t) - r \|\dot\beta^*(t)\|$, and hence, $\|\dot\beta^*(t) / \dot\beta^*_0(t)\| \le 1/ r$.

\relax
\section{Numerical Experiments}
\label{s:simu1}

\subsection{A small experiment with a triangular $\scX$.}

\begin{figure}[htp]
\includegraphics[scale = 0.4]{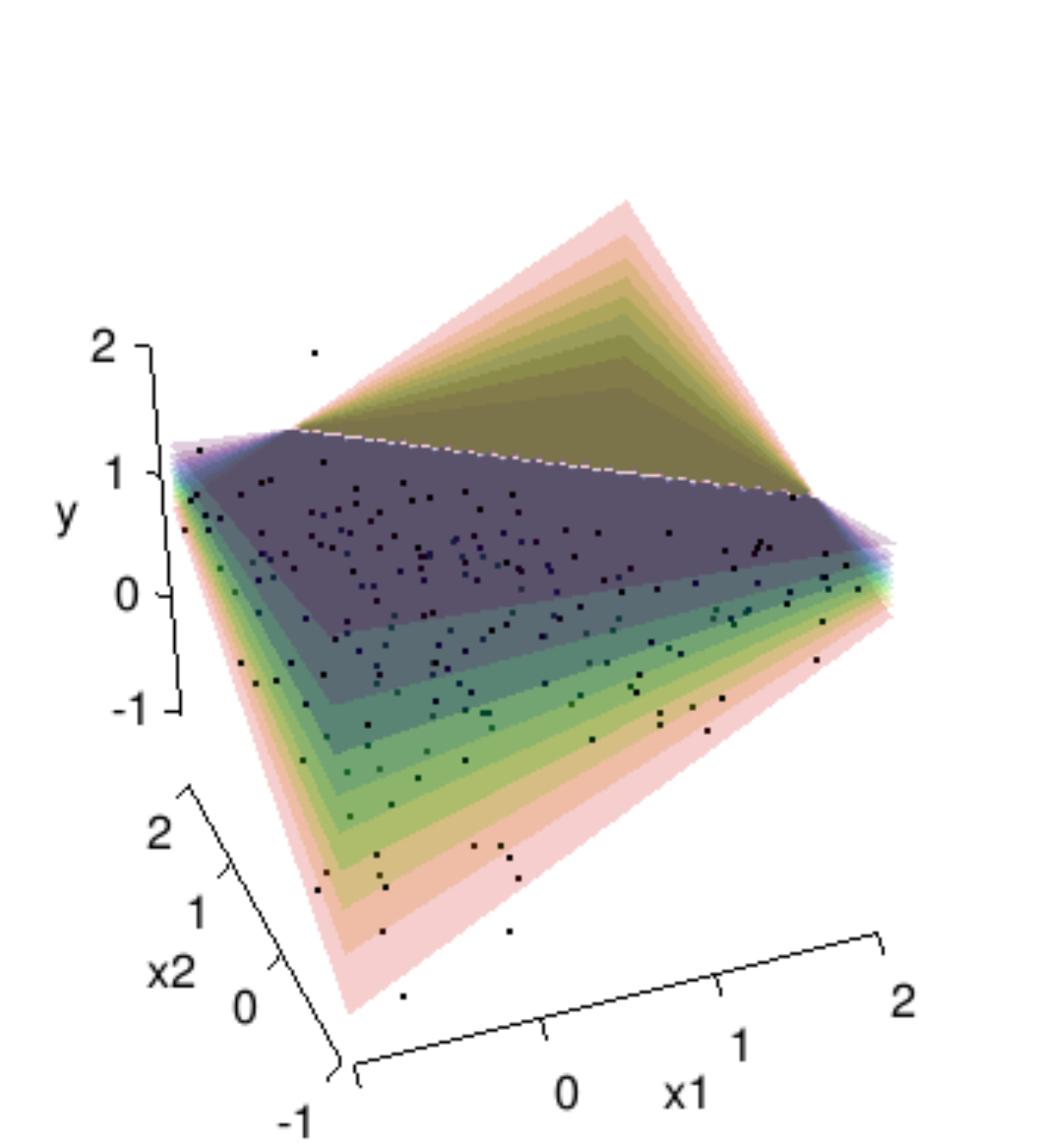}
\includegraphics[scale = 1]{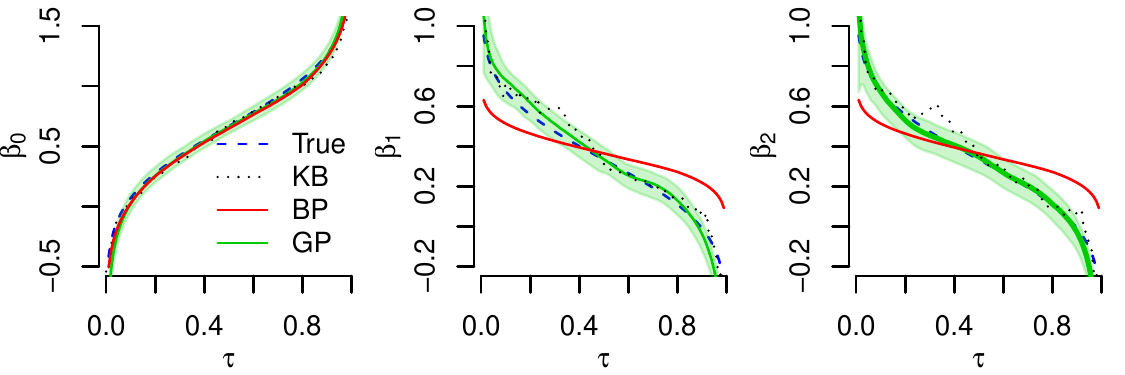}
\caption{Triangular $\scX$ example. Left panel shows quantile planes and their extensions to the embedding rectangle. Other three panels show estimated coefficient curves. KB: classical Koenker-Basette estimates; BP: Bayesian estimates from Bernstein polynomial model of \citet{Reich2011}; GP: estimates from the proposed Gaussian process method.}\label{f:tri}
\end{figure}

To illustrate why adjusting to the shape of $\scX$ is important for joint QR estimation, we generated 200 synthetic observations from the model:
\begin{align}
& X \sim \mbox{Uniform}(\scX);~~\scX = \{x = (x_1, x_2)^T  \in \bbR^2: -1 \le x_1, x_2 \le 2, x_1 + x_2 \le 1\},\nonumber\\
& Q_Y(\tau|X) =  \frac{1 - (X_1 + X_2)}{3}Q_N(\tau | 0,1) + \frac{2 + X_1 + X_2}{3}Q_N(\tau | 1, 0.2^2) \label{triQ}
\end{align}
where $Q_N(\tau|\mu, \sigma^2)$ denotes the $\tau$-th quantile of the $N(\mu, \sigma^2)$ distribution. The hyperplanes on the right hand side of \eqref{triQ} are correctly ordered on the triangular predictor space $\scX$, but cross each other inside the smallest embedding rectangle $[-1,2]\times[-1,2]$, as seen on the left panel of Figure \ref{f:tri}. This negatively impacts estimation by the \citet{Reich2011} method (Figure \ref{f:tri}), which cannot adapt to the triangular shape of the predictor space and is restricted to estimates that do not cross on the smallest rectangle enclosing all observed predictors. In contrast, our method, which works on the convex hull of the observed predictors, can retrieve the true parameter curves with a much higher accuracy.

\subsection{Performance assessment: univariate $X$}

For a thorough study of the frequentist performances of the proposed method, we simulated $100$ synthetic datasets each with $n = 1000$ observations from the model
\[
X \sim \mbox{Uniform}(-1,1);~~~Q_Y(\tau|X) =  3(\tau-\tfrac12) \log \frac{1}{\tau(1-\tau)} + 4(\tau-\tfrac12)^2 \log\frac1{\tau(1-\tau)}X,
\]
and compared parameter estimation against the methods of \citet{Reich2011} and \citet{Koenker1978}. Here $X$ is one dimensional, and the shape of $\scX$ is a not an issue. However, the nearly quadratic $\beta_1(\tau)$ function is slightly challenging to estimate.  

\begin{figure}[htp]
\centering
\includegraphics[scale = .9]{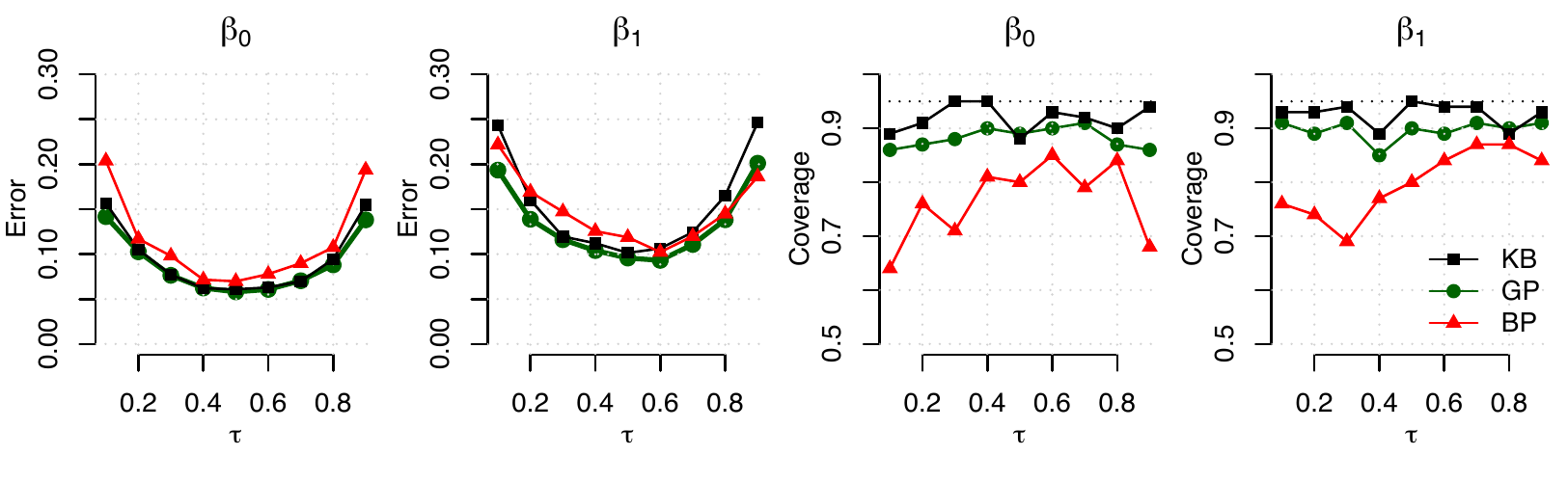}
\caption{Assessing performance with univariate $X$. Left two panels show mean absolute estimation errors and right two panels show coverage by 95\% confidence or credible bands. KB: classical Koenker-Basette estimates; BP: Bayesian estimates from Bernstein polynomial model of \citet{Reich2011}; GP: estimates from the proposed Gaussian process method.}
\label{f:simu1}
\end{figure}

We used the default setting for the method by \citet{Reich2011} with 5 basis functions. For each implementation, the Gibbs sampler was run for 10000 iterations and 200 samples from the second half of the chain were used for Monte Carlo. We also tried two other versions with 10 and 15 basis functions respectively. But increasing the number of basis functions resulted in a progressively poor performance, and thus we only report here the results from the 5 basis function setting. The \texttt{QuantReg} package in \texttt{R} was used to implement the classical method by \citet{Koenker1978} and confidence intervals were constructed with 200 bootstrapped samples. For our Gaussian process method, we used 6 equispaced knots $\tau^*_k = (k-1)/5$, $k = 1,\ldots,6$. We ran the adaptive blocked Metropolis sampler for 10000 iterations with 10\% burn-in and used 200 samples from the rest for Monte Carlo. Nearly identical results were obtained with 11 equispaced knots. 

Figure \ref{f:simu1} shows comparisons of pointwise mean absolute estimation errors of the three methods and also the coverage of the associated 95\% confidence or credible bands, averaged across the 100 synthetic datasets. Our method offered lowest estimation errors over the entire range of $\tau$ values, and a consistently high coverage close to the nominal target of 95\%. It is important to remember that the credible bands produced by our method are calibrated in a Bayesian way, and so 95\% credible bands are not automatically guaranteed to offer 95\% coverage.

\subsection{Performance assessment: multivariate $X$}

For assessing performance in the multivariate case, we ran another simulation study with synthetic data generated from the model:
\[
X \sim \mbox{Uniform}(\{x \in \mathbb{R}^7: \|x\| \le 1\});~~~Q_Y(\tau|X) =  \beta_0(\tau) + X^T\beta(\tau)
\]
with $\beta_0$ and $\beta$ specified by the equations
\begin{align*}
&\beta_0(0.5) = 0,~~~\beta(0.5) = \begin{pmatrix}0.96 &  -0.38 &   0.05 &  -0.22 &  -0.80 &  -0.80 &   -5.97\end{pmatrix}^T,\\
& \dot\beta_0(\tau)  = \frac{1}{\tau(1-\tau)};~~\dot\beta(\tau) = \frac{\dot\beta_0(\tau)v(\tau)}{\sqrt{1 + \|v(\tau)\|^2}},~~\tau \in (0,1),
\end{align*}
where $v_j(\tau) = \sum_{l = 0}^2 a_{lj} \cdot \phi(\tau; {l}/{2}, 1/{(3^2)})$, $1\le j \le 7$, with $\phi(\cdot|\mu,\sigma^2)$ denoting the $N(\mu,\sigma^2)$ density function and 
\[
a = \begin{pmatrix} 0 & 0 & -3 & -2 & 0 & 5 & -1\\
-3 & 0 & 0 & 2 & 4 & 1 & 0\\
0 & -2 & 2 & 2 & -4& 0 & 0
\end{pmatrix}.
\]
These specifications define a valid model by Theorem \ref{thm:char} because $a(b, \scX) = 1$ for any non-zero $b$ when $\scX$ is the unit ball centered at zero. Also note that with these specifications, $Q_Y(\tau|0)$ is precisely the quantile function of the standard logistic distribution. For simulating an $(X,Y)$ from the model we set $X = U_1 Z/\|Z\|$ and $Y = Q_Y(U_2|X)$ where $U_1,U_2 \sim U(0,1)$ and $Z \sim N_7(0,I_7)$, drawn independently of each other. We evaluated each instance of $Q_Y(U_2|X)$ to a precision of $10^{-16}$ by numerically integrating $\dot\beta_0$ and $\dot\beta$ between $0.5$ and $U_2$ with the \textsf{integrate()} function in \textsf{R}.
\begin{figure}[!t]
\centering
\includegraphics[scale = .9]{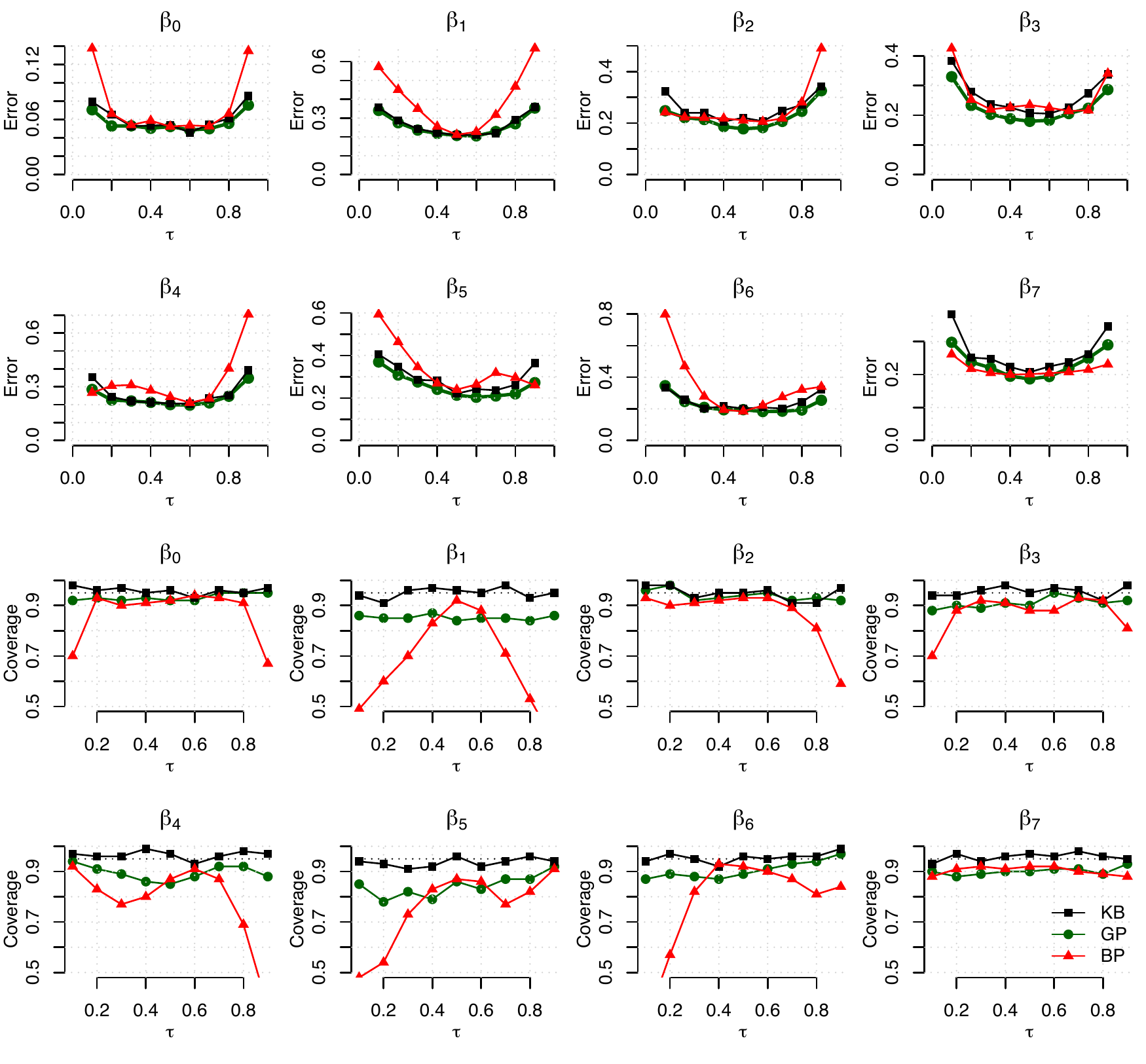}
\caption{Assessing performance with a 7 dimensional $X$. Top two rows show estimation errors and bottom two rows show coverage by 95\% confidence or credible bands. KB: classical Koenker-Basette estimates; BP: Bayesian estimates from Bernstein polynomial model of \citet{Reich2011}; GP: estimates from the proposed Gaussian process method.}
\label{f:simu7}
\end{figure} 

Figure \ref{f:simu7} compares the estimation error and coverage of the three methods averaged across 100 datasets of size $n = 1000$ generated from the above simulation model. All three methods were set as in the previous example. Like before, our method again offered nearly lowest estimation errors and a consistently high coverage close to the nominal target of 95\%,  over the entire range of $\tau$ values.

\relax
\section{Case studies}
\label{s:cs}
\subsection{Plasma concentration of beta-carotene}

\cite{nierenberg} presents a study of the association of beta-carotene plasma concentrations with dietary intakes and drugs use for nonmelanoma skin cancer patients. The Statlib database (\url{http://lib.stat.cmu.edu/datasets/Plasma_Retinol}) hosts a subset of the data from 315 patients who had an elective surgical procedure during a three-year period to biopsy or remove a lesion of the lung, colon, breast, skin, ovary or uterus that was found to be non-cancerous. This dataset has been analyzed in the literature \citep{kai2011new} to assess how personal characteristics, smoking and dietary habits as well as dietary intake of beta-carotene affects concentration levels of beta-carotene in the plasma. 

\begin{figure}[!ht]
\centering
\includegraphics[scale = .8]{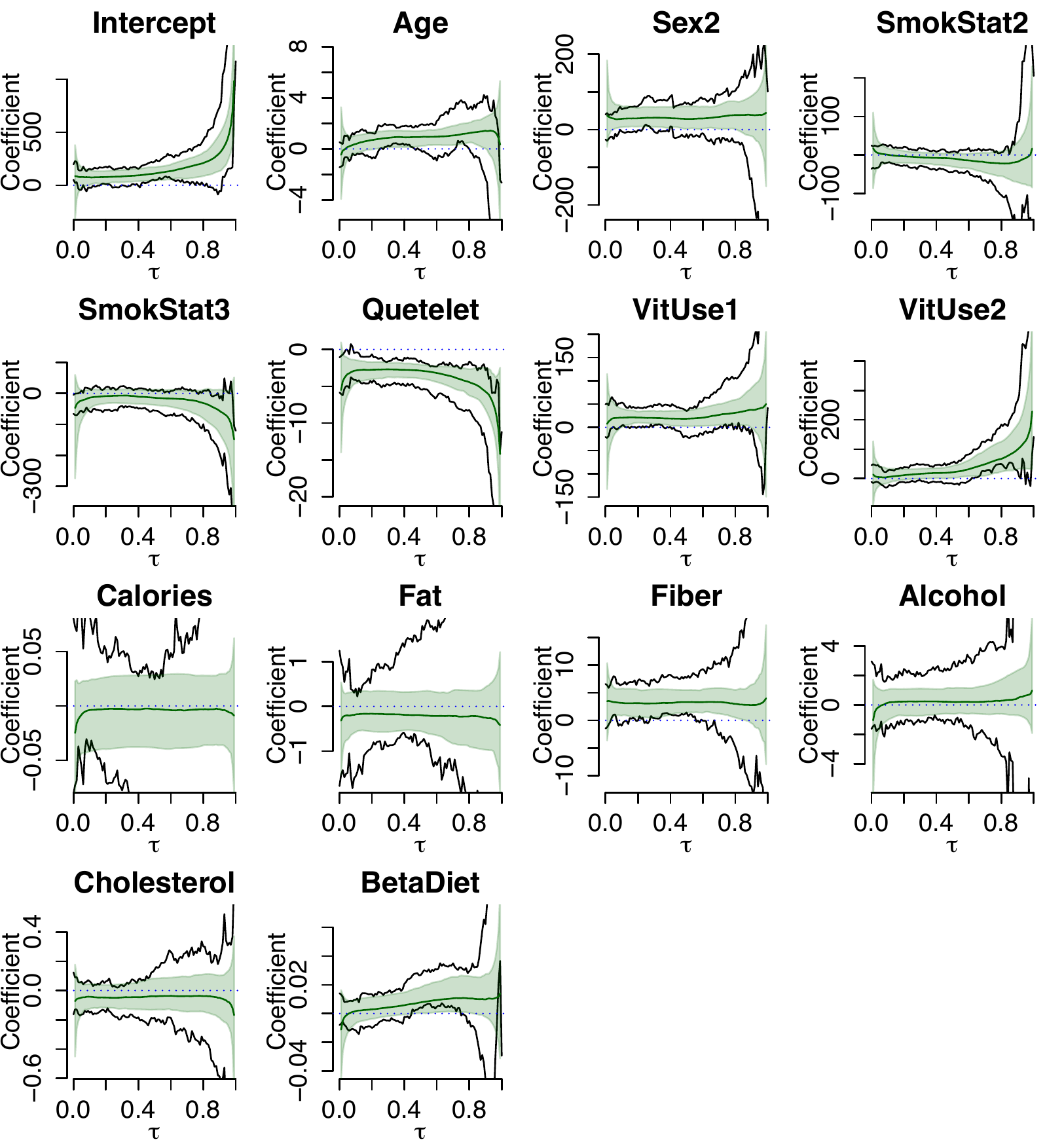}
\caption{Parameter estimation for plasma data analysis. Green lines and bands give posterior means and pointwise 95\% credible bands for intercept and slope curves, overlaid with the 95\% bootstrap confidence bands obtained from single-$\tau$ Koenker-Bassett fits shown in black.}
\label{f:plasma}
\end{figure}

We analyzed the same data with our joint QR model with plasma beta-carotene concentration (ng/ml) as the response and 11 covariates consisting of age (years), sex (1=Male, 2=Female), smoking status (1=Never, 2=Former, 3=Current Smoker), Quetelet index or BMI (weight/(height$^2$)), vitamin use\footnote{relabeled for better clarity as `$3~-$ the original label'} (1=No, 2=Yes, not often, 3=Yes, fairly often), and daily consumption of calories, fat (g), fiber (g), alcohol (number of drinks), cholesterol (mg) and dietary beta-carotene (mcg). These covariates gave rise to 13 predictors when the categorical variables (sex, smoking status and vitamin use) were coded with dummy indicators. Estimated intercept and slope curves, with 95\% credible bands are shown in Figure \ref{f:plasma}.

The estimated intercept curve strongly suggests a longer right tail for the response distribution. The slope curve estimates indicate that being female, use of vitamin and consumption of fiber have reasonably strong positive effect on plasma concentration of beta-carotene, whereas, smoking and BMI have reasonably strong negative effect. Calories, fat, alcohol or cholesterol consumption appears to have little effect. Dietary intake of beta-carotene appears to have a positive effect, but the inference is not conclusive. The slope estimates in Figure \ref{f:plasma} suggest more dramatic effects of some predictors on the upper quantiles, but the credible bands paint a more modest picture. However, credible bands for $\beta_j(0.9) - \beta_j(0.1)$ and $\beta_j(0.9) - \beta_j(0.5)$, constructed directly from the posterior draws, indeed suggest more enhanced positive and negative effects, respectively for heavy vitamin use and BMI, on the upper quantiles (Table \ref{t:plasma}).

We also performed a ten fold validation study to assess how well our joint model captured the intricacies of the beta-carotene data. In each fold of the study, we randomly partitioned the 315 observations into training and test sets at roughly 2:1 ratio. We fitted our joint model on the training data and obtained estimates $\hat\beta_j$ of $\beta_j$, $j = 0,\ldots,p$ in the form of posterior means. These estimates were then used to evaluate the training and test data ``check'' loss at every $\tau \in \{0.1, \ldots, 0.9\}$ by averaging $\rho_\tau(Y_i - \hat\beta_0(\tau) - X_i^T\hat\beta(\tau))$ over, respectively, all training and all test set observations $(X_i, Y_i)$, where $\rho_\tau(r) = r\{\tau -I(r < 0)\}$. The same was done with Koenker-Bassette, \cite{Reich2011} and standard least squares estimates. The relative accuracy of a method at any $\tau$ was calculated as the reciprocal of its check loss at that $\tau$ relative to the least square method. Figure \ref{f:plasma-eff} shows these relative accuracy measures for the three quantile regression methods, averaged across the 10 repetitions. Our joint QR method can be seen to offer the best test data accuracy across all $\tau$ values and maintain its advantage over least squares at the upper quantiles where the other two quantile regression methods appear to suffer a sharp loss of efficiency.

\begin{table}
\centering
\begin{tabular}{rlcc}
$j$ & Predictor & 95\% CI for $\beta_j(0.9) - \beta_j(0.1)$ & 95\% CI for $\beta_j(0.9) - \beta_j(0.5)$\\
\hline
1  &  Age  & $( -0.52 ,  2.32 )$ & $( -0.88 ,  1.69 )$\\
2  &  Sex2  & $( -43.34 ,  54.12 )$ & $( -33.53 ,  55.31 )$\\
3  &  SmokStat2  & $( -66.81 ,  28.51 )$ & $( -47.64 ,  35.7 )$\\
4  &  SmokStat3  & $( -102.75 ,  26.64 )$ & $( -95.05 ,  19.97 )$\\
\bf5  &  \bf Quetelet  & $( -5.43 ,  0.12 )$ & $\mathbf{( -4.93 ,  -0.14 )}$\\
6  &  VitUse1  & $( -21.05 ,  54.43 )$ & $( -14.1 ,  55.09 )$\\
\bf 7  &  \bf VitUse2  & $\mathbf{( 12.27 ,  177.02 )}$ & $\mathbf{( 10.57 ,  153.73 )}$\\
8  &  Calories  & $( -0.01 ,  0.01 )$ & $( -0.01 ,  0.02 )$\\
9  &  Fat  & $( -0.75 ,  0.33 )$ & $( -0.58 ,  0.25 )$\\
10  &  Fiber  & $( -3.7 ,  3.14 )$ & $( -3.22 ,  2.54 )$\\
11  &  Alcohol  & $( -0.74 ,  2.32 )$ & $( -0.72 ,  1.51 )$\\
12  &  Cholesterol  & $( -0.12 ,  0.16 )$ & $( -0.12 ,  0.09 )$\\
13  &  BetaDiet  & $( 0 ,  0.02 )$ & $( -0.01 ,  0.01 )$
\end{tabular}
\caption{Evidence of more dramatic upper tail effects of certain predictors on plasma beta-carotene concentration. CI denotes posterior credible interval.}
\label{t:plasma}
\end{table}

\begin{figure}
\centering
\includegraphics[scale = .6]{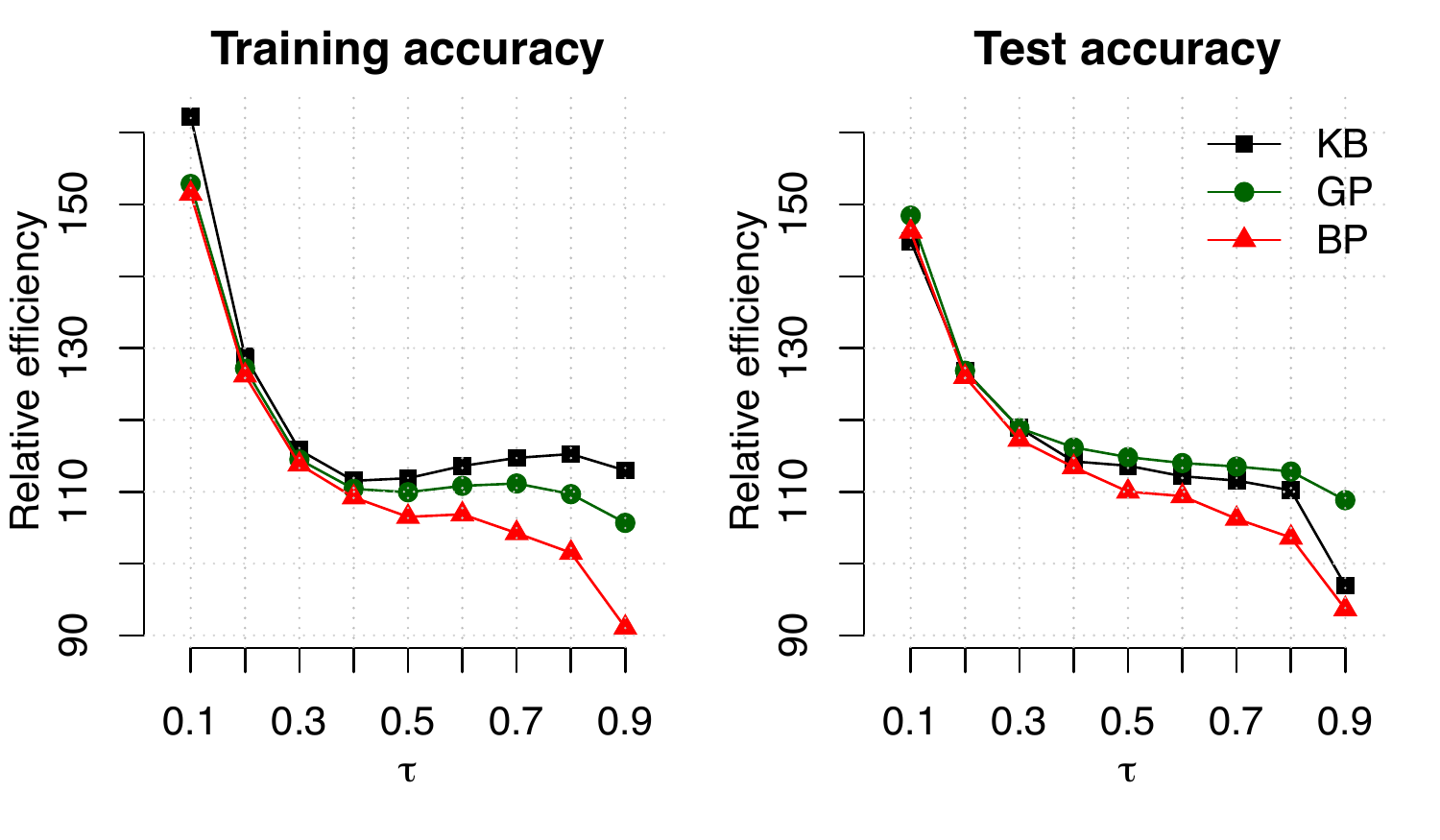}
\caption{A 10-fold cross validation assessment of the fit of various linear quantile regression methods to plasma data with standard least squares regression being the benchmark. In held out test data, the proposed Gaussian process method (GP)  offers better fit at all quantiles than Koenker-Bassette (KB) or the method by \citet[BP]{Reich2011}.}
\label{f:plasma-eff}
\end{figure}

\subsection{Survival analysis under right censoring}
Joint estimation of quantile regression parameters could be particularly beneficial for survival analysis with censored response. A greater borrowing of information may help cover the information gaps left by censoring. A crossing-free estimation of the quantile functions means that the estimated survival curves are proper and interpretable. Also, a joint estimation offers an automatic way to quantify estimation uncertainty of the entire survival curves by simple inversions of estimated quantile functions. The probabilistic modeling framework of our joint quantile regression approach makes it particularly straightforward to handle right-censoring. The log-likelihood score calculation \eqref{eq8} now changes to
\begin{align}\label{eq8c}
&\sum_{i} [(1 - c_i) \log f_Y(y_i|x_i) + c_i \log\{1 - F_Y(y_i|x_i)\}] \nonumber\\
&~~~~~~~~~~~~~~~=\sum_i \left[c_i \log\{1 - \tau_{x_i}(y_i)\}-(1 - c_i) \log\big\{\dot{\beta}_0\big(\tau_{x_i}(y_i)\big)+
 x_i^T\dot{\beta}\big(\tau_{x_i}(y_i)\big)\big\}\right],
\end{align}
where $c_i$ is the censoring status (1= right censored, 0 = observed). With this single change, the same prior specification and Markov chain Monte Carlo parameter estimation as detailed in Section \ref{s:prior} remain applicable.

We illustrate these points with a reanalysis of the University of Massachusetts Aids Research Unit IMPACT Study data \citep[UIS,][Table 1.3]{Hosmer1998} in which we estimated the conditional quantiles of the logarithm of the time to return to drug use ($Y$) as linear functions of current treatment assignment (\textsf{TREAT}, 1 = Long course, 0 = Short course), number of prior drug treatments (\textsf{NDT}), recent intravenous drug use (\textsf{IV3}, 1 = Yes, 0 = No), Beck depression score (\textsf{BECK}), a compliance factor measuring length of stay in the treatment relative to the course length (\textsf{FRAC}), race of the subject (\textsf{RACE}, 1 = Non-white, 0 = White), age (\textsf{AGE}) and treatment site (\textsf{SITE}). For model fitting, we used the 575 complete observations available in the \textsf{uis} data set of the R package \textsf{quantreg}. Return times were right censored for 111 of these subjects.

Figures \ref{f:uis-coef}-\ref{f:uis-surv} show parameter and survival curves (for 9 randomly chosen subjects) estimation with our joint quantile regression approach and also with the censored quantile regression approach described in \citet{Koenker2008}. The latter was implemented by using the \textsf{crq} function in R-package \textsf{quantreg} which uses a technique by \citet{Portnoy2003}. For joint estimation, we fixed the base probability density $f_0$ to be $N(0,1)$ instead of a $t_\nu$, since the tails of the distribution of log return time are expected to be fast decaying. The two sets of parameter estimates are comparable, except in the upper tails. The Portnoy method fails to produce an estimate beyond $\tau = 0.88$, and confidence intervals get extremely wide for $\tau$ close to this limit. In contrast, the credible bands from joint estimation are much more stable across the entire range of $\tau$. Estimated survival curves are remarkably similar, though for the Portnoy method, the issue of quantile crossing manifests in the form of estimated survival curves that are not strictly decreasing (e.g., subject \# 313).

\begin{figure}[!t]
\centering
\includegraphics[scale = .6]{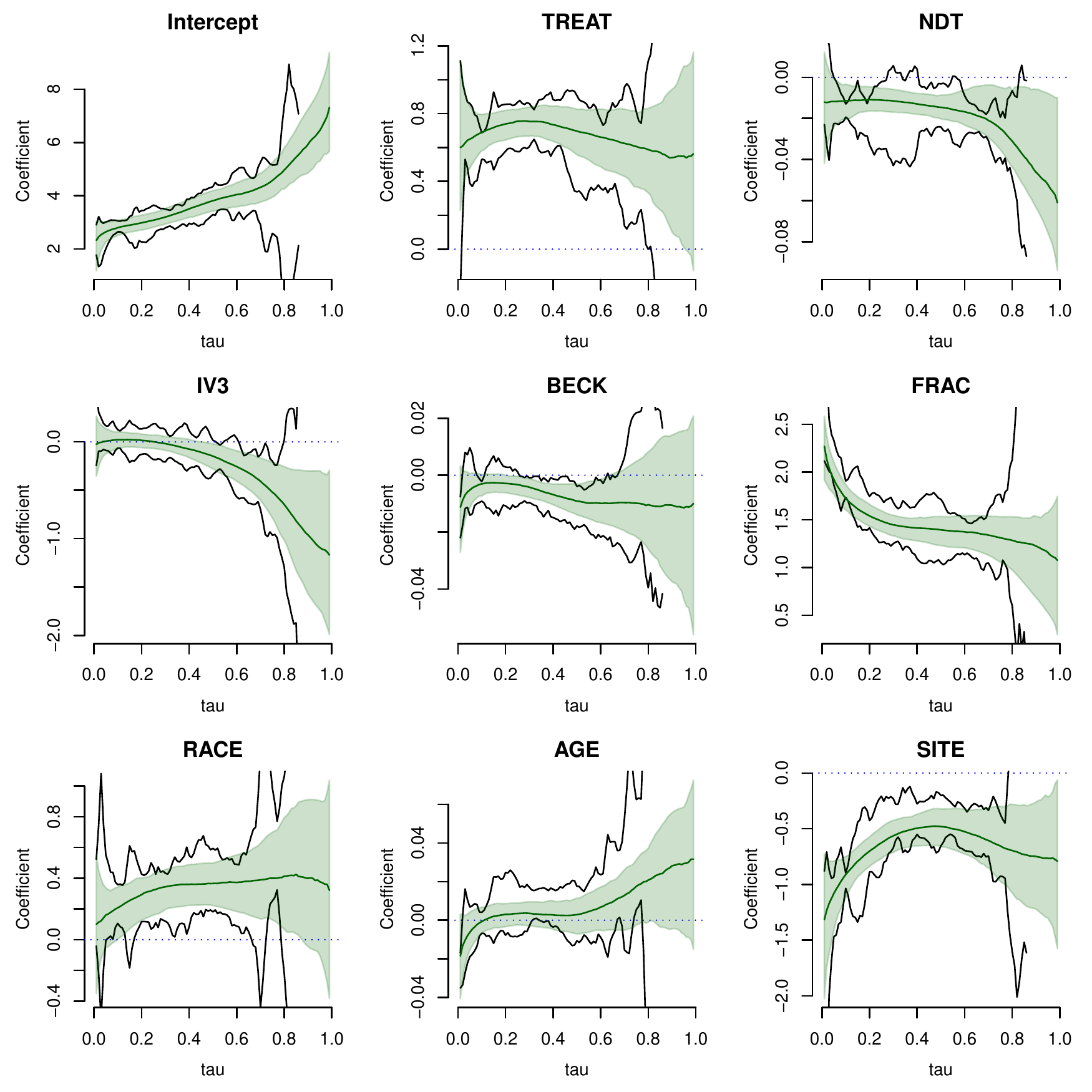}
\caption{Parameter estimation for UIS data analysis. Green lines and bands give posterior means and pointwise 95\% credible bands for intercept and slope curves, overlaid with the 95\% bootstrap confidence bands obtained from the Portnoy approach.}
\label{f:uis-coef}
\end{figure} 
\begin{figure}[!t]
\centering
\includegraphics[scale = .6]{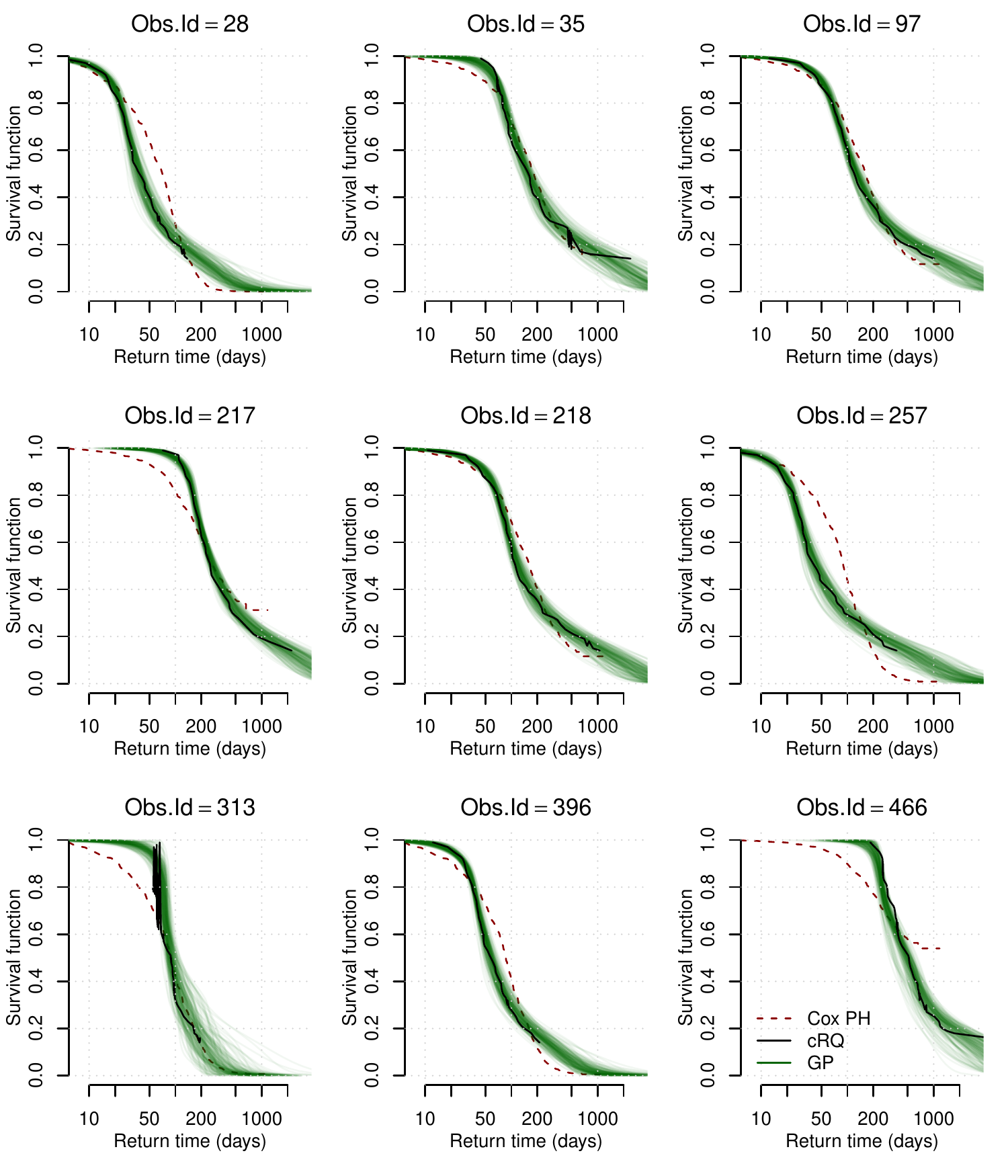}
\caption{Estimated survival curves for 9 random sampled subjects in the UIS study. Green lines are posterior draws of the survival curves. Blacks lines are the estimates from the Portnoy approach, and the dark red dashed lines are estimates under the Cox proportional hazard model.}
\label{f:uis-surv}
\end{figure} 
\relax
\section{Discussion}
\label{s:dis}

We have introduced a complete and practicable theoretical framework for simultaneous estimation of linear quantile planes in any dimension and over arbitrarily shaped convex predictor domains. Although we have pursued here a specific estimation procedure, our modeling platform is extremely broad and parameter estimation could be done in a variety of other manners. For example, one could choose to use spline based estimation of the basic functions $w_0, \ldots, w_p$ via penalized likelihood maximization or Bayesian averaging. Also, a variety of specifications could be used on the diffeomorphism parameter $\zeta$, e.g., one could model $\zeta$ as a mixture of beta cumulative distribution functions, or try estimating $\zeta$ directly by adding isotonic regression type constraints.

%
%

A number of interesting features could be added to the Bayesian parameter estimation method we have pursued here. An important consideration is shrinkage for large $p$. For moderately large $p$, any standard shrinkage prior could be used on $\gamma$, and the resulting posterior could be explored by the same Markov chain sampler as in Section \ref{s:comp} as long as the prior density on $\gamma$ is available in an explicit form up to a normalizing constant. Shrinkage could also be applied on the curve valued parameters $w_j$, $j = 1, \ldots, p$, by choosing appropriate prior distributions on $(\kappa^2_1, \ldots, \kappa^2_p)$. An attractive choice is to replace the single gamma prior distribution we used in Section \ref{s:prior} with a spike-slab type mixture of gamma distributions, e.g., $\kappa_j^{-2} \sim 0.5 Ga(a_\kappa, b_\kappa) + 0.5 Ga(a_\kappa, b_\kappa/100)$. Such a specification still allows integrating out $\kappa_j$ in \eqref{eq:pp} and hence could be explored by the same Markov chain sampler as before. 


The primary computational bottleneck of our method is that the likelihood evaluation involves a search over the grid of $\tau$ values for each observation. While our current implementation easily scales to thousands of observations, scaling it to even larger datasets will require further computing innovations. Fortunately, the likelihood evaluation is embarrassingly parallel in the observations and involves very simple arithmetic operations, and thus, it should be possible to obtain manyfold speed ups by the use of graphics processing units; such an implementation is currently underway.


%

\begin{appendices}
\section{Technical details}
This section presents a proof of Theorem \ref{th:2}, starting with a few fundamental results that allow comparing two probability density functions given information on their corresponding quantile density functions. We adopt the following notation in the remainder of this section: by a `probability function quartet' we mean a four-tuple $(Q, q, F, f)$ of real valued functions where $Q:(0,1) \to \bbR$ is a non-atomic quantile function that admits a strictly positive derivative $q = \dot Q$, $F = Q^{-1}$ is the associated cumulative distribution function with probability density function  $f = \dot F$. Recall the identities $f(y) = 1/q(F(y))$ and $q(t) = 1/f(Q(t))$.

\subsection{Auxiliary results}
\begin{lemma}\label{lem:aux2}
Let $(Q_1, q_1, F_1, f_1)$ and $(Q_2, q_2, F_2, f_2)$ be two probability function quartets and take $m_j = Q_j(\tau_0)$, $j = 1,2$. If there exist $0 < c_1 \le 1 \le c_2 < \infty$ such that $c_1 q_2(t) \le q_1(t) \le c_2 q_2(t)$, for all $t \in (0,1)$, then, 
\[
f_2(y) = f_1(y + \Delta_1(y)) \cdot \Delta_2(y),~~\mbox{for all}~y \in {\rm supp}(f_2),
\]
for two real valued functions $\Delta_1, \Delta_2$ satisfying $|\Delta_1(y)| \le \max(1 - c_1, c_2 - 1) |y - m_2| +  |m_1 - m_2|$, and, $\Delta_2(y) \in [c_1, c_2]$.
\end{lemma}

\begin{proof}
By the assumption on $q_1, q_2$, for every $y \in {\rm supp}(f_2)$, 
\[
c_1 y + b_1 \le Q_1F_2(y) \le c_2 y + b_2,
\]
where $b_1 = m_1 - c_1m_2$, $b_2 = m_1 - c_2m_2$. Then, $\Delta_1(y) := Q_1F_2(y) - y$ satisfies, $\|\Delta_1(y)\| \le \max(1 - c_1, c_2 - 1)|y - m_2| + |m_1 - m_2|$. Since, $f_i(y) = 1/q_i(F_i(y))$, $i = 1,2$, we have, for any $y \in {\rm supp}(f_2)$, $f_1(Q_1(F_2(y)))q_1(F_2(y))  = 1$, and hence,
\[
f_2(y) =  \frac{f_1(Q_1(F_2(y)))q_1(F_2(y))}{q_2(F_2(y))} = f_1(y + \Delta_1(y)) \cdot \frac{q_1(F_2(y))}{q_2(F_2(y))},
\]
which proves the result. 
\end{proof}

\begin{lemma}\label{lem:aux3}
Let $f^*$ satisfy the conditions of Theorem \ref{th:2}. Given any $\delta, \sigma, c_1, c_2 > 0$, there exists an $\epsilon > 0$ such that ,
\[
\sup_{x \in \scX} \int_{[Q^*_Y(\epsilon|x), Q^*_Y(1 - \epsilon|x)]} f^*_Y(y|x) \left|\log \{ f_0(y /\sigma + \Delta(y))/\sigma\}\right|dy < \delta
\]
for every $\Delta: \bbR \to \bbR$ satisfying $|\Delta(y)| < c_1|y| + c_2$ for all $y \in \bbR$.
\end{lemma}

\begin{proof}
By the assumption on $f^*$, $q^*_Y(t|x) / q^*_Y(t|0) = 1 + x^T\dot\beta^*(t)/\dot\beta^*_0(t)$ is bounded away from zero and infinity. Hence, by Lemma \ref{lem:aux2}, there are constants $a, b > 0$ such that for every $x \in \scX$, 
$Q^*_Y(F^*_Y(y|0)|x) = y + \Delta_{1,x}(y)$, with $|\Delta_{1,x}(y)| \le a|y| + b$ for all $y \in A_0 := {\rm supp}(f^*_Y(\cdot|0))$. Fix any $x \in \scX$. By the change of variable $z = Q^*_Y(F^*_Y(y|x)|0)$,
\begin{align*}
\int f^*_Y(y|x) & \left|\log f_0\left(\frac{y}\sigma + \Delta(y)\right)\right|dy\\
& = \int_{A_0} f^*_Y(z|0)\left|\log f_0\left(\frac{Q^*_Y(F^*_Y(z|0)|x)}\sigma + \Delta(Q^*_Y((F^*_Y(z|0)|x))\right)\right|dz \\
& = \int_{A_0} f^*_Y(z|0) \left|\log f_0(z/\sigma + \Delta_2(z))\right|dz,
\end{align*}
with $|\Delta_2(z)| \le |\Delta_1(z)|/\sigma + |\Delta(z + \Delta_1(z))| \le a_1|z| + b_1$ where $a_1,b_1 > 0$ depend only on $\delta$, $\sigma$, $c_1$ and $c_2$. The tail assumption on $f^*_Y(\cdot|0)$ implies that the last integral is finite, proving the result!
\end{proof}

\begin{lemma}
Fix $\gamma_0 \in \bbR$, $\gamma \in \bbR^p$, $\sigma > 0$, $w : (0,1) \to \bbR^p$, and two differentiable, monotonically increasing functions $\zeta, \dagg\zeta : [0,1] \to [0,1]$, with $[\zeta(0), \zeta(1)] \subset [\dagg\zeta(0), \dagg\zeta(1)]$. Let $(\beta_0, \beta) = \mathcal{T}(\gamma_0, \gamma, \sigma, w, \zeta)$, $(\dagg\beta_0, \dagg\beta) = \mathcal{T}(\dagg\gamma_0, \dagg\gamma, \sigma, w, \dagg\zeta)$, where,
\[
\dagg\gamma_0 = \gamma_0 + \sigma \int_{\zeta(\tau_0)}^{\dagg\zeta(\tau_0)} q_0(u) du,~~\dagg\gamma = \gamma + \sigma\int_{\zeta(\tau_0)}^{\dagg\zeta(\tau_0)} q_0(u) h(u) du,
\]
with $h(\tau) := w(\tau)/\{a(w(\tau), \scX) \sqrt{1 + \|w(\tau)\|^2}\}$, $\tau \in (0,1)$. Fix any $x \in \scX$ and consider the probability function quartets $(Q_Y(\cdot|x), q_Y(\cdot|x), F_Y(\cdot|x), f_Y(\cdot|x))$, $(\dagg Q_Y(\cdot|x), \dagg q_Y(\cdot|x), \dagg F_Y(\cdot|x), \dagg f_Y(\cdot|x))$ where $Q_Y(\tau|x) = \beta_0(t) +x^T\beta(\tau)$, $\dagg Q_Y(\tau|x) = \dagg\beta_0(\tau) + x^T\dagg\beta(\tau)$. Then, $ {f_Y(y|x)}/{\dagg f_Y(y|x)} = {\dagg{\dot\zeta}(\dagg F_Y(y|x))}/{\dot\zeta(F_Y(y|x))}$ for all $y \in {\rm supp}(f_Y(\cdot|x))$.
\label{lem:aux}
\end{lemma}

\begin{proof}
Let $\tau_1 = \zeta(\tau_0)$ and, define,
\[
q_0(\tau|x) = q_0(\tau) \{1 + x^T h(\tau)\},~~\tau \in (0,1), x \in \scX.
\]
Then $q_0(\tau|x)$ is strictly positive, and, hence,
\[
Q_0(\tau|x) := \int_{\tau_1}^{\tau} q_0(u|x) du,~~ \tau \in (0,1), x \in \scX,
\]
defines valid quantile planes on $\scX$. Denote the associated conditional distribution and density functions by $F_0(\cdot|x)$ and $f_0(\cdot|x)$. By definition of $(\beta_0, \beta)$, $q_Y(\tau|x) = \dot \beta_0(\tau) + x^T \dot\beta(\tau) = \sigma q_0(\zeta(\tau))\left\{1 + x^T h(\zeta(\tau))\right\}\dot\zeta(\tau) = \sigma q_0(\zeta(\tau)|x)\dot\zeta(\tau)$, and, hence,
\begin{equation}
Q_Y(\tau|x) = \gamma_0 + x^T \gamma + \int_{\tau_0}^\tau q_Y(u|x)du = \gamma_0 + x^T \gamma + \sigma Q_0(\zeta(\tau)|x)\label{eq:a1}.
\end{equation}
Similarly, $\dagg q_Y(\tau|x) = \sigma q_0(\dagg\zeta(\tau)|x)\dot{\dagg\zeta}(\tau)$, and, hence,
\begin{equation}
\dagg Q_Y(\tau|x) = \dagg\gamma_0 + x^T \dagg\gamma + \int_{\tau_0}^\tau \dagg q_Y(u|x)du = \gamma_0 + x^T \gamma + \sigma Q_0(\dagg\zeta(\tau)|x).
\end{equation}
by the definitions of $\dagg\gamma_0$, $\dagg\gamma$.
Inverting \eqref{eq:a1}, we get, for every $x \in \scX$,
\[
\zeta(\tau) = F_0\left(\frac{Q_Y(\tau|x) - \gamma_0 - x^T \gamma}{\sigma} \bigg| x \right), \tau \in (0,1).
\]
Therefore, if $y \in (Q_Y(0|x), Q_Y(1|x))$, then,
\[
 f_Y(y|x) = \frac{1}{q_Y(F_Y(y|x) | x)} = \frac{1}{\sigma q_0(\zeta(F_Y(y|x))|x) \dot \zeta(F_Y(y|x))} = \frac{f_0(\frac{y - \gamma_0 - x^T \gamma}{\sigma} | x)}{\sigma \dot \zeta(F_Y(y|x))}.
\]
Similarly, $\dagg f_Y(y|x) = f_0(\frac{y - \gamma_0 - x^T \gamma}{\sigma} | x) / \{\sigma \dot{\dagg\zeta}(\dagg F_Y(y|x))$\}, proving the result!
 \end{proof}

\subsection{Approximating $f^*$ within assumed model space}
\label{a:proof4}
Let $\Pi_\nu$ denote the conditional prior distribution on $f$ under $\Pi$ given $\nu$. Theorem \ref{th:2} is proved in two stages. Let $\nu_0 > 0$ such that the tails of $f^*_Y(\cdot|0)$ are of type I or II with respect to $f_0(\cdot|\nu)$ for every $0 < \nu \le \nu_0$. First we show that for any such $\nu$ and any given $\delta > 0$, there exists an $\dagg{f} \in K_\delta(f^*)$ within our model space with nicely behaved underlying $w_j$ curves. Next we show $\Pi_\nu(f:\|\log (\dagg{f} / f)\|_\infty < \delta) > 0$ which leads to the claim of Theorem \ref{th:2}. The following lemma gives a precise statement of the first step.

\begin{lemma}\label{le:6}
Let $f^*$ satisfy the conditions of Theorem \ref{th:2}. For any small $\delta>0$ and $0 < \nu \le \nu_0$, there exists an $\dagg{f} \in K_\delta(f^*)$ associated with $(\dagg\beta_0, \dagg\beta) = \mathcal{T}(\dagg\gamma_0, \dagg\gamma, \dagg\sigma, \dagg w, \dagg\zeta)$ where $\dagg w: [0,1] \to \bbR^p$ is bounded continuous and  $\dagg \zeta: [0,1] \to [0,1]$ is a diffeomorphism with $\dot{\dagg\zeta}(t) \in [e^{-B}, e^B]$ for all $t \in [0,1]$ for some finite $B > 0$.
\end{lemma}

\begin{proof}
Fix a $\nu \in (0, \nu_0)$ and a $\delta \in (0, \tau_0)$. All calculations below are carried out for this particular value of $\nu$ and we suppress $\nu$ from the notation $f_0(\cdot|\nu)$. 

Let $\gamma_0^* = \beta^*_0(\tau_0)$. Fix a $\sigma_L > 0$ such that $c_L(\sigma_L) \le 1/2$ if $f^*_Y(\cdot|0)$ has a type I left tail with respect to $f_0$, or, $u_L(\sigma_L) \log \{1/u_L(\sigma_L)\} \le \delta/2$ if the left tail is of type II. Similarly fix $\sigma_R$ and take $\sigma^* = \min\{\sigma_L, \sigma_R\}$. Define $\zeta^* : [0,1] \to [0,1]$ as
\[
\zeta^*_0(t) = F_0\left(\gamma^*_0 + \frac{\beta^*_0(t) - \gamma^*_0}{\sigma^*} \right), t \in [0,1],
\]
which is differentiable and monotonically increasing, and, whose derivative can be written as,
\[
\dot\zeta^*_0(t) = \frac1{\sigma^*}f_0\left(\gamma^*_0 + \frac{\beta^*_0(t) - \gamma^*_0}{\sigma^*}\right)\dot\beta^*_0(t) = \frac{(1/\sigma^*)f_0\left(\gamma^*_0 + \frac{\beta^*_0(t) - \gamma^*_0}{\sigma^*}\right)}{f^*_Y(\beta^*_0(t)|0)},~t \in (0,1),
\]
since $\beta^*_0(t) = Q_Y^*(t|0)$. 

Because $\zeta^*$ has a continuously differentiable inverse on $[\zeta^*(0), \zeta^*(1)]$, the relation $h^*(\zeta^*(u)) = \dot \beta^*(u) / \dot \beta^*_0(u)$ defines a map $h^*: [\zeta^*(0), \zeta^*(1)] \to \bbR^p$ that is bounded and continuous by the assumption of Theorem \ref{th:2}, and hence, can be extended to a bounded continuous function $h^*:[0,1] \to \bbR^p$. Define $w^* :[0,1] \to \bbR^p$ as follows, essentially repeating the construction in the ``Only if part'' of the proof of Theorem \ref{thm:char}. If $h^*(t) = 0$ then set $w^*(t) = 0$. Otherwise, take $c(t) = [ [1/\{\|h^*(t)\|a(h^*(\tau), \scX)\}]^2  - 1]^{-1/2}$ and set $w^*(t) = c(t) h^*(t) / \|h^*(t)\|$. By the assumption on $\dot\beta^*/\dot\beta^*_0$, $w^*$ is a bounded continuous function on $[0,1]$. 

By construction $(\beta^*_0, \beta^*) = \mathcal{T}(\gamma^*_0, \gamma^*, \sigma^*, w^*, \zeta^*)$. However this parameter vector may not be in our model space since we may have either $[\zeta^*(0) , \zeta^*(1)] \ne [0,1]$ or $\|\dot \zeta^*\|_\infty = \infty$. We correct this by introducing a proper diffeomorphism $\dagg\zeta$ on $[0,1]$ with $\dot{\dagg\zeta}$ bounded away from 0 and infinity, such that $\dagg\zeta(t) = \zeta^*(t)$ for $t \in [\delta_L, 1 - \delta_R]$ for suitably chosen small numbers $\delta_L, \delta_R > 0$. This is the crux of the approximation argument. 

If the left tail is type I, then $\zeta^*(0) > 0$ and $\lim_{t \downarrow 0} \dot \zeta^*(t) = c_L(\sigma^*) \in (0,1/2]$. So one can fix $\delta_L > 0$ small enough such that $\zeta^*(\delta_L) > \delta_L$, $\dot\zeta^*(t) \in (c_L(\sigma^*)/2, 1]$ for all $t \in (0, \delta_L]$, and, $\delta_L \log [4 / \{c_L(\sigma^*)\delta_L\}] < \delta/2$. Otherwise, the left tail is type II, and in that case choose $\delta_L = u_L(\sigma^*)$, which automatically ensures $\dot \zeta^*(t) \ge 1$ for all $t \in (0,1)$ with $\dot\zeta^*(\delta_L) = 1$, and, $\delta_L\log(1/\delta_L) \le \delta/2$. Since $\zeta^*(0) \ge 0$, we also must have $\zeta^*(\delta_L) \ge \delta_L$. Fix $\delta_R$ by repeating the same steps with the right tail. Define $\dagg\zeta: [0,1] \to \bbR$ as,
\[
\dagg\zeta(t) = \left\{ \begin{array}{ll} \zeta^*(t), & t \in [\delta_L,1 - \delta_R],\\[5pt] a_L t^2 + b_L t,& t \in [0, \delta_L) \\[5pt] 1 - a_R (1 - t)^2 - b_R (1 - t), & t \in (1 - \delta_R,1],\end{array}\right.
\]
where, 
\begin{align*}
& a_L = \frac{\delta_L \dot \zeta^*(\delta_L) - \zeta^*(\delta_L)}{ \delta_L^2},~~b_L = \frac{2 \zeta^*(\delta_L) - \delta_L \dot\zeta^*(\delta_L)}{\delta_L},\\
& a_R = \frac{\delta_R \dot \zeta^*(1 - \delta_R) - \{1 - \zeta^*(1 - \delta_R)\}}{ \delta_R^2},~~b_R = \frac{2 \{1 - \zeta^*(1 - \delta_R)\} - \delta_R \dot\zeta^*(1 - \delta_R)}{\delta_R}.
\end{align*}
By choice of $\delta_L$ and $\delta_R$, $a_L < 0$, $a_R < 0$ and $b_L \in [\dot\zeta^*(\delta_L), 2/\delta_L]$, $b_R \in [\dot\zeta^*(1 - \delta_R), 2 / \delta_R]$. It is straightforward to verify that $\dagg\zeta$ defines a diffeomorphism from $[0,1]$ onto $[0,1]$, with $\dagg{\dot \zeta}(t) \in [\dot\zeta^*(\delta_L), b_L]$ for all $t \in [0,\delta_L]$ and $\dagg{\dot \zeta}(t) \in [\dot\zeta^*(1 - \delta_R), b_R]$ for all $t \in [1 - \delta_R, 1]$. Therefore there exists a $B > 0$ such that $\dot {\dagg \zeta}(t) \in [e^{-B}, e^{B}]$ for all $t \in [0,1]$. 

Take $(\dagg\beta_0, \dagg\beta)=\mathcal{T}(\gamma^*_0, \gamma^*, \sigma^*, w^*, \dagg\zeta)$ with valid conditional quantile planes $\dagg Q_Y(\cdot|x)$ and associated cumulative distribution and probability density functions given by $\dagg F_Y(\cdot|x)$ and $\dagg f_Y(\cdot|x)$. By construction of $\dagg\zeta$, $\dagg Q_Y(\tau|x) = Q^*_Y(\tau|x)$ for all $\tau \in [\delta_L, 1 - \delta_R]$ and hence $\dagg F_Y(y|x) = F^*_Y(y|x)$ for all $y \in [Q^*_Y(\delta_L|x), Q^*_Y(1 - \delta_R|x)]$. Hence, by Lemma \ref{lem:aux}, 
\[
d_{KL}(f^*_Y(\cdot|x), \dagg f_Y(\cdot|x)) = \int_{y \in Q^*_Y(\delta_L|x), Q^*_Y(1 - \delta_R|x)]^c} f^*_Y(y|x) \log \frac{\dagg{\dot\zeta}(\dagg F_Y(y|x))}{\dot\zeta^*(F^*_Y(y|x))} dy.
\]
Split the integral above into two integrals, one over $y <  Q^*_Y(\delta_L|x)$ and the other over $y >  Q^*_Y(1 - \delta_R|x)$. When $y < Q^*_Y(\delta_L|x)]$, both $\dagg F_Y(y|x) < \delta_L$, and,  $F^*_Y(y|x) < \delta_L$. Clearly, $\dot{\dagg\zeta}(\dagg F_Y(y|x)) \le b_L \le 2 / \delta_L$. If left tail is type I then, $\dot \zeta^*(F^*_Y(y|x)) \ge c_L(\sigma^*) / 2$ and hence,
\begin{align*}
\int_{Q^*_Y(0|x)}^{Q^*_Y(\delta_L|x)} f^*_Y(y|x) \log \frac{\dagg{\dot\zeta}(\dagg F_Y(y|x))}{\dot\zeta^*(F^*_Y(y|x))} dy & \le  \left[\log \frac{4}{\delta_L c_L(\sigma^*)}\right]\int_{Q^*_Y(0|x)}^{Q^*_Y(\delta_L|x)} f^*_Y(y|x)dy\\
& = \delta_L \log \frac{4}{\delta_L c_L(\sigma^*)} \le \delta/2
\end{align*}
by the choice of $\delta_L$ for the type I left tail. On the other hand, if the left tail is type II, then $\dot\zeta^*(F^*_Y(y|x)) \ge 1$ and hence 
\[
\int_{Q^*_Y(0|x)}^{Q^*_Y(\delta_L|x)} f^*_Y(y|x) \log \frac{\dagg{\dot\zeta}(\dagg F_Y(y|x))}{\dot\zeta^*(F^*_Y(y|x))} dy \le \delta_L \log \frac{2}{\delta_L} \le \delta/2,
\]
again by the choice of $\delta_L$ for this case. Same arguments apply to the integral over $y \in [Q^*_Y(1 - \delta_R|x), Q^*_Y(1|x)]$, and hence, for every $x \in \scX$, $d_{KL}(f^*_Y(\cdot|x), \dagg f_Y(\cdot|x)) \le \delta$. Therefore $d_{KL}(f^*, \dagg f) = \int f_X(x) d_{KL}(f^*_Y(\cdot|x), \dagg f_Y(\cdot|x))dx \le \delta$. 
\end{proof}


\subsection{Proof of Theorem \ref{th:2}}
Since the prior on $\nu$ has full support, it suffices to show that given any $\delta > 0$ and $\nu < \nu_0$, the conditional prior $\Pi_\nu = \Pi(\cdot|\nu)$ assigns positive mass to the event $\{f:d_{KL}(f^*, f) < 3\delta\}$. Fix any $\delta > 0$, and $\nu < \nu_0$. By Lemma \ref{le:6}, there is a $(\dagg\beta_0, \dagg\beta) = \mathcal{T}(\dagg\gamma_0, \dagg\gamma, \dagg\sigma, \dagg w, \dagg\zeta)$ with the associated probability density function $\dagg f$ satisfying $d_{KL}(f^*, \dagg f) < \delta$, where, $\dagg w: [0,1] \to \bbR^p$ is bounded continuous, and, $\dagg\zeta :[0,1] \to [0,1]$ is a diffeomorphism with $\|\log \dot{\dagg\zeta}\|_\infty < \infty$. 

For any $\lambda > 0$, let $A_\lambda$ denote the set of $(\gamma_0, \gamma, \sigma, w, \zeta)$ such that $|\gamma_0 - \dagg\gamma_0| < \lambda$, $\|\gamma - \dagg\gamma\| < \lambda\dagg \sigma / {\rm diam}(\scX)$, $|\sigma /\dagg\sigma - 1| < \lambda$, $w:[0,1] \to \bbR^p$ is continuous and $\sup_t \|w(t) - \dagg w(t)\| < \lambda$, $\zeta: [0,1] \to [0,1]$ is a diffeomorphism and $\|\log \dot\zeta - \log \dagg{\dot\zeta}\|_\infty < \lambda$. By construction and because of the full support properties of Gaussian processes \citep{tokdar&ghosh07}, the conditional prior $\Pi_\nu$ assigns a positive mass to the set of $f$ associated with $(\beta_0, \beta) = \mathcal{T}(\gamma_0, \gamma, \sigma, w, \zeta)$, $(\gamma_0, \gamma, \sigma, w, \zeta) \in A_\lambda$, for every $\lambda > 0$. So, it suffices to show that $\lambda > 0$ could be chosen small enough such that any $f$ associated with a $(\gamma_0, \gamma, \sigma, w, \zeta) \in A_\lambda$ satisfies $\int f^*_Y(y|x) \log \{\dagg f_Y(y|x)/ f_Y(y|x)\} dy \le 2\delta$ for all $x \in \scX$. 

Let $b = \|\log\dot{\dagg\zeta}\|_\infty + 1$. Since $\dagg w$ is bounded on $[0,1]$, there exists a $B > 0$ such that $\dagg h(t) := {\dagg w(t) }/\{a(\dagg w(t), \scX)\sqrt{1 + \|\dagg w(t)\|^2}\}$ satisfies
\[
1 + x^T\dagg h(t) \in \left[e^{-B}, e^B\right],~\mbox{for all}~t \in [0,1].
\]
Clearly, there exists a $\lambda_0 \in (0, 1/2)$ such that $\|\log\dot\zeta- \log\dot{\dagg\zeta}\|_\infty < \lambda_0$ implies $\|\log\dot\zeta\|_\infty < 2b$, and, $\sup_t \|w(t) - \dagg w(t)\| < \lambda_0$ implies $1 + x^T w(t) / \{a(w(t), \scX)\sqrt{1 + \|w(t)\|^2}\} \in [e^{-2B}, e^{2B}]$ for all $t \in [0,1]$. Take $c_1 = 1 + e^{2B}$, $c_2 = c_1\{|\dagg\gamma_0| + \|\dagg\gamma\| +1)\} + 1/2$, and, $\tilde c_1 = (1 + 2c_1)/\dagg\sigma$, $\tilde c_2 = 2(c_2 + |\dagg\gamma_0| + 1)/\dagg\sigma$. By Lemma \ref{lem:aux3} there exists an $0 < \epsilon < \delta / \max\{12b, 6(B+1)\}$ such that, 
\[
\sup_{x \in \scX} \int_{[Q^*_Y(\epsilon|x), Q^*_Y(1 - \epsilon|x)]} f^*_Y(y|x) |\log f_0(y/\dagg\sigma + \Delta(y))/\dagg\sigma| dy < \delta/3.
\] 
for every $\Delta: \bbR \to \bbR$ satisfying $|\Delta(y)| < \tilde c_1 |y| + \tilde c_2$ for all $y \in \bbR$

Take any $(\gamma_0, \gamma, \sigma, w, \zeta) \in A_{\lambda_0}$ and let $(\beta_0, \beta) = \mathcal{T}(\gamma_0, \gamma, \sigma, w, \zeta)$, $(\beta^{e\dagger}_0, {\beta}^{e\dagger}) = \mathcal{T}(\gamma^{e\dagger}_0, \gamma^{e\dagger}, \dagg\sigma, \dagg w, e)$, and, $(\beta^e_0, \beta^e) = \mathcal{T}(\gamma^e_0, \gamma^e, \sigma, w, e)$, where $e$ denotes the identity function on $[0,1]$ onto itself and
\begin{align*}
\gamma^{e\dagger}_0 &= \dagg\gamma_0 + \dagg\sigma \int_{\dagg\zeta(\tau_0)}^{\tau_0} q_0(u)du,~~\gamma^{e\dagger} = \dagg\gamma + \dagg\sigma \int_{\dagg\zeta(\tau_0)}^{\tau_0} q_0(u)\dagg h(u) du,\\
\gamma^{e}_0 &= \gamma_0 + \sigma \int_{\zeta(\tau_0)}^{\tau_0} q_0(u)du,~~\gamma^{e} = \gamma + \sigma \int_{\zeta(\tau_0)}^{\tau_0} q_0(u) h(u) du.\\
\end{align*}
with $h(t) = w(t) / \{a(w(t), \scX)\sqrt{1 + \|w(t)\|^2}\} $, $t \in [0,1]$. The definitions of $\gamma_0, \gamma^e_0, \gamma^{e\dagger}_0, \gamma^{e\dagger}$ match the requirements of Lemma \ref{lem:aux}. Let $(Q_Y(\cdot|x), q_Y(\cdot|X), F_Y(\cdot|x), f_Y(\cdot|x))$ denote the probability function quartet of $(\beta_0, \beta)$, and the same symbols with appropriate superscripts denote the same quantities associated with the other three pairs $(\beta^e_0, \beta^e)$, $(\dagg\beta_0, \dagg\beta)$ and $(\beta^{e\dagger}_0, \beta^{e\dagger})$. 

Consider the following factorization in log-scale
\[
\log \frac{\dagg f_Y(y|x)}{f_Y(y|x)} = \log \frac{\dagg f_Y(y|x)}{f^{e\dagger}_Y(y|x)} + \log \frac{f^e_Y(y|x)}{f_Y(y|x)} + \log \frac{f^{e\dagger}_Y(y|x)}{f^{e}_Y(y|x)}.
\]
By Lemma \ref{lem:aux}, $|\log \{f^{\dagger}_Y(y|x) / f^{e\dagger}_Y(y|x)\}| = |\log \dot{\dagg\zeta}(\dagg F_Y(y|x))| \le 2b$, and, $|\log\{f^e_Y(y|x) / f_Y(y|x)\}| = |\log \dot\zeta(F_Y(y|x))|\le 2b$. 
Since $1 + x^T h(t) \in [e^{-2B}, e^{2B}]$ for all $t \in [0,1]$, we have, for any $x \in \scX$, $q^e_Y(t|x) = q^e_Y(t|0) \cdot [e^{-2B}, e^{2B}]$ for all $t \in (0,1)$, and hence, by Lemma \ref{lem:aux2}, $f^e_Y(y|x) = f^e_Y(y + \Delta_{1,x}(y)|0) \Delta_{2,x}(y)$, with $|\Delta_{1,x}(y)| \le c_1 |y| + c_2$ for all $y \in \bbR$, and, $\|\log \Delta_{2,x}\|_\infty \le 2B$. But, $f^e_Y(y|0) = f_0((y - \gamma_0) / \sigma)/\sigma$ and so, $f^e_Y(y|x) = f_0(y/\dagg\sigma + \tilde \Delta_{1,x}(y))\tilde \Delta_{2,x}(y)/\dagg\sigma$ with $|\tilde \Delta_{1,x}(y)| \le \tilde c_1 |y| + \tilde c_2$, for all $y \in \bbR$, and, $\|\log \tilde \Delta_{2,x}\|_\infty \le B + 1$. The same calculations work for $f^{e\dagger}_Y$ because $(\dagg\gamma_0, \dagg\gamma, \dagg\sigma, \dagg w, \dagg\zeta) \in A_{\lambda_0}$. Therefore, 
\[
\int_{[Q^*_Y(\epsilon|x), Q^*_Y(1 - \epsilon | x)]^c} f^*_Y(y|x) \log \frac{\dagg f_Y(y|x)}{f_Y(y|x)}dy < \delta,
\]
for every $x \in \scX$.

The map $(x, y) \mapsto \log \dagg f_Y(y|x)$ is equicontinuous on $\{(x,y): x \in \scX, y \in [Q^*_Y(\epsilon|x), Q^*_Y(1 - \epsilon|x)]\}$, and hence there exist a $\kappa > 0$ such that $\log|\dagg f_Y(y+z|x) / \dagg f_Y(y|x)| < \delta/2$ for all $x \in \scX$, $y \in [Q^*_Y(\epsilon|x), Q^*_Y(1 - \epsilon|x)]$, $|z| < \kappa$. Fix a small $0 < \eta < \kappa/2$ such that 
\[
\max(e^\eta - 1, 1 - e^{-\eta})\cdot\sup_{x\in\scX} \max\{|Q^*_Y(\epsilon|x) - \dagg\gamma_0 - x^T\dagg\gamma|, |Q^*_Y(1 - \epsilon|x) - \dagg\gamma_0 - x^T\dagg\gamma|\} < \frac\kappa2.
\]
By the equicontinuity of the maps $s \mapsto \log q_0(e^s)$ and $s \mapsto \dagg h(e^s)$ on the interval $[\log \epsilon, \log (1 - \epsilon)]$, and the continuity of the transformation $v \mapsto v / \{a(v, \scX)\sqrt{1 + \|v\|^2}\}$, one can fix $0 < \lambda < \min(\lambda_0, \kappa/4)$ such that for any $(\gamma_0, \gamma, \sigma, w, \zeta) \in A_\lambda$,
\[
\frac{q_Y(t|x)}{\dagg q_Y(t|x)} = \frac\sigma{\dagg\sigma} \times \frac{q_0(\zeta(t))}{q_0(\dagg\zeta(t))} \times \frac{1 + x^Th(\zeta(t))}{1 + x^T\dagg h(\dagg\zeta(t))} \times \frac{\dot \zeta(t)}{\dot{\dagg\zeta}(t)} \in [e^{-\eta}, e^{\eta}],
\]
for every $t \in [\epsilon,1-\epsilon]$ and $x \in \scX$. Consequently, by Lemma \ref{lem:aux2} for every $x \in \scX$ and $y \in [Q^*_Y(\epsilon|x), Q^*_Y(1 - \epsilon|x)]$, $|\log \{f_Y(y|x) / \dagg f_Y(y|x)\}| < \delta$. This proves the result.


\section{Computational details}


\IncMargin{1em} 
\begin{algorithm}
\SetKwInOut{Input}{input}
\SetKwInOut{Output}{output}
\Input{Model parameters\\
scalars: $\gamma_0$, $\sigma$, \\
vectors: $\gamma = (\gamma_1, \cdots, \gamma_p)^T$, \\
functions: $w = (w_1, \cdots, w_p):(0,1) \to \bbR^p$, diffeomorphism $\zeta$ on $[0,1]$\\
}
\Output{Log-likelihood score}
\tcp{Basic quantities}
\lFor{$l = 1:m$}{
set $\dot b_{0,l} = \sigma q_0(\zeta(t_l)) \dot\zeta(t_l)$\;
}
\For(~~~~\texttt{// Could be parallelized in $l$}){$l = 1:m$}{set $v_l = \omega(\zeta(t_l))$\;
\lFor{$i = 1:n$}{
calculate $a_{i} = x_i^T v_l$\;
}
calculate $a_{\scX} = \max_{1 \le i \le n} \{-a_{i}\} / \sqrt{\|v_l\|}$\;
\lFor{$i = 1:n$}{
set $\tilde a_{il} = a_i / \{a_{\scX} \cdot \sqrt{1 + \|v_l\|^2}\}$\;
}
}
\tcp{Calculatelog likelihood score by sequencing through obs}
set $\ell\ell = 0$ \tcp*[l]{initialize the log likelihood}
\For(~~~~\texttt{// Could be parallelized in $i$}){$i = 1:n$}{
calculate $Q_0 = \gamma_0 + \gamma^T x_i$\;
\eIf{$Y_i > Q_0$}{
set $Q_U = Q_0$ and $l = k$\;
\While{$Y_i > Q_U$}{
set $Q_L = Q_U$ and $l = l + 1$\;
\lIf{$l \le m$}{calculate $~~~~Q_U = Q_L + (\delta/2) \cdot \{\dot b_{0,l-1}(1 +  \tilde a_{i,l-1}) + \dot b_{0,l}(1 + \tilde a_{i, l})\}$}\; \lElse{set $Q_U = \infty$}
}
}{
set $Q_L = Q_0$ and $l = k + 1$\;
\While{$Y_i \le Q_L$}{
set $Q_U = Q_L$ and $l = l - 1$\;
\lIf{$l \ge 2$}{calculate $~~~~Q_L = Q_U - (\delta/2) \cdot \{\dot b_{0,l-1}(1 +  \tilde a_{i,l-1}) + \dot b_{0,l}(1 + \tilde a_{i, l})\}$}\; \lElse{set $Q_L = -\infty$}\;
}}
\eIf{$Q_L = -\infty$ {\bf or} $Q_U = \infty$}{set $\ell\ell = -\infty$}{
calculate $\alpha = (Y_i - Q_L) / (Q_U - Q_L)$\;
set $\ell\ell = \ell\ell - \log\{(1 - \alpha) \dot b_{0,l-1}(1 +  \tilde a_{i,l-1}) + \alpha \dot b_{0,l}(1 + \tilde a_{i, l})\}$}
}
\Return{$\ell\ell$}
\caption{Log-likelihood evaluation}\label{algo:loglik} \end{algorithm}
\DecMargin{1em}

\subsection{Centering the predictors}\label{A:centering}
A preprocessing step of our method is to center the observed predictors $\{x_1, \ldots, x_n\}$ around an interior point of their convex hull (Figure \ref{fig:centering}). While the sample mean vector automatically gives an interior point, it may lie too close to the hull boundary and lead to poorer model fit. A better strategy is to use the mean of the extreme points of the data cloud, but finding the extreme points becomes computationally intensive for $p > 2$. Instead, we employ a fast algorithm that recursively identifies $p + 1$ points $x^*_1, \ldots, x^*_{p+1}$, from the data cloud that are close to the boundary and far away from each other.


Consider a Gaussian process $f$ on $\bbR^p$ with covariance function $C(x, x') = \exp\{-\|\Delta^{-1}(x - x')\|^2\}$, where $\Delta$ is the $p\times p$ diagonal matrix with $j$-th element equaling the observed range of the $j$-th predictor, $j = 1, \ldots,p$. Take $x^*_1 = x_1$ and recursively select $x^*_j$ as the $x \in \{x_1, \ldots, x_n\}$ with maximum ${\rm Var}(f(x) | x^*_1, \ldots, x^*_{j-1})$, $j = 2, \ldots, p+1$. This recursive selection can be carried out extremely fast, with computational complexity of the order $(p+1)n\log n$ flops, by carrying out a rank-$(p+1)$ incomplete, pivoted Cholesky decomposition of the $n\times n$ non-negative definite matrix $K = ((C(x_i, x_j)))$, for example, by using the {\sf inchol} function of the {\sf R} package {\sf kernlab}. Such implementations depend on the order in which the $x_i$s are stored. To encourage selection close from the boundary, we prearrange the $x_i$s in decreasing order of their Mahalanobis distance $\|S^{-1}(x_i - \bar x)\|$ from mean $\bar x$, where $S$ denotes the sample covariance.

\begin{figure}
\centering
\includegraphics[scale =0.85]{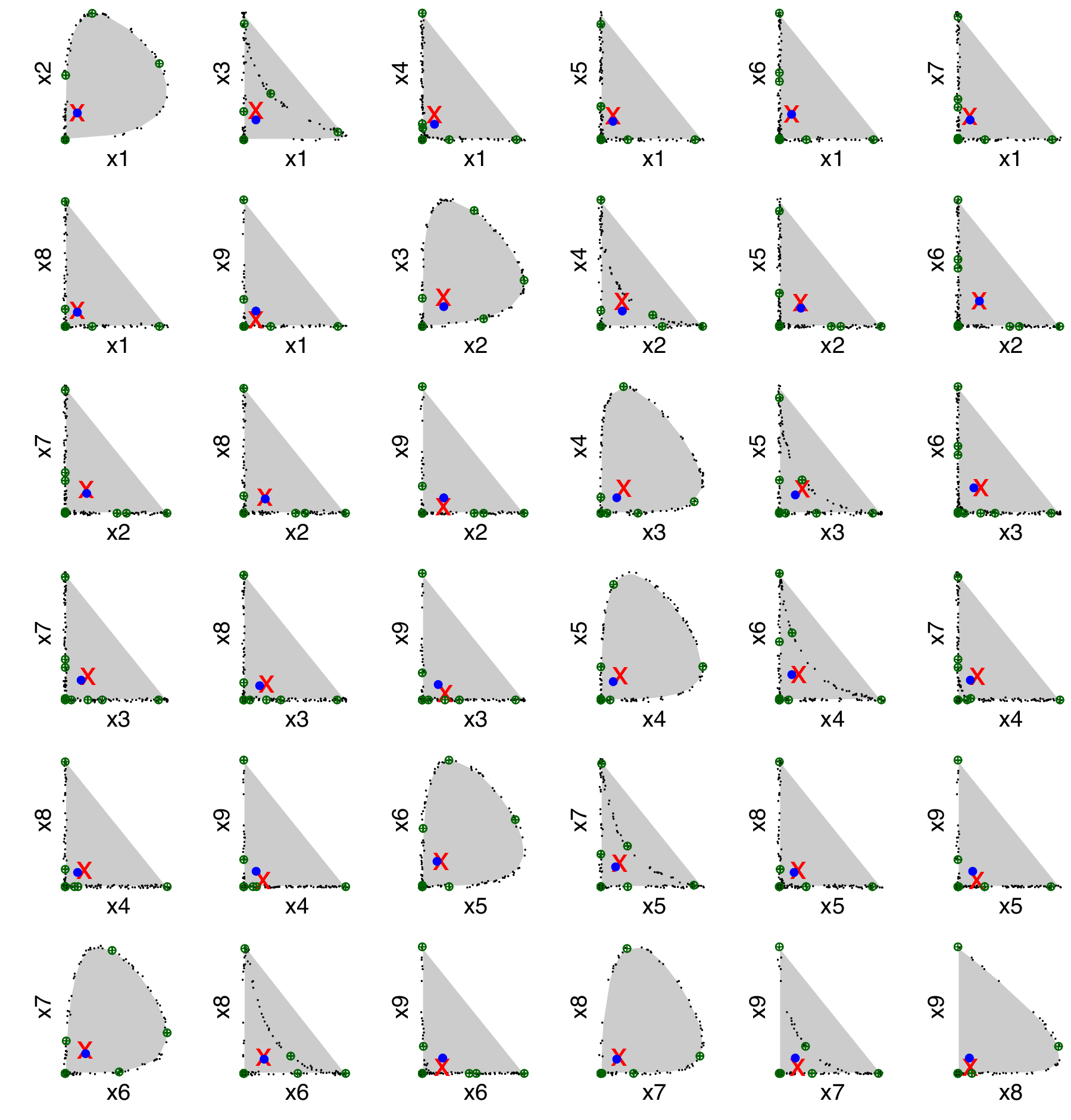}
\caption{A toy demonstration of finding an interior point of the convex hull of observed predictors. The observed predictor vectors are  9 dimensional B-spline transforms of 500 uniform random drawn from $[0,1]$, shown as black dots (with some jitter to improve visibility) in the pairwise plots. The red X denotes the projection of the sample mean, and the large blue dot denotes the projection of the interior point found by our preprocessing method. The green crosshairs are the selected $10$ points. }\label{fig:centering}
\end{figure}

\subsection{Choosing $\lambda$ grid points}\label{a:support}
In choosing the grid points $\lambda_g$, $g = 1, \ldots, G$, for $\lambda_j$, it is important to ensure that the conditional prior distributions $N(0, \kappa_j^2 C_{**}(\lambda_g))$ remain sufficiently overlapped for neighboring $\lambda_g$ values, since otherwise, the grid based discretization of the prior on $\lambda$ may lead to poor mixing of the Markov chain sampler. If overlap is measured by the Kullback-Leibler divergence $d(\lambda, \lambda') := d_{KL}(N(0, \kappa_j^2 C_{**}(\lambda)), N(0, \kappa_j^2 C_{**}(\lambda')))$, which does not depend on $\kappa_j$, it is easy to see that one must use a non-uniform grid of $\lambda$ values since for a given $\Delta > 0$, $d(\lambda, \lambda + \Delta)$ is much larger for a small $\lambda$ than a large one. To choose this non-uniform grid, we set $\lambda_1$ to be the smallest value in the predetermined range, one that gives $\rho_{0.1}(\lambda_1) = 0.99$, and then increment $\lambda$ recursively so that $d(\lambda_{g-1}, \lambda_{g}) = 1$, $g = 2, 3, \ldots$, until the whole range is covered. 

\end{appendices}

\bibliographystyle{chicago}
\bibliography{qr}

\end{document}